\documentclass[preprint,5p]{elsarticle}

\renewcommand{\today}{February 28, 2011}
\newwrite\lafont
\immediate\openout\lafont=la14.mf%
\catcode`#=11
\immediate\write\lafont{%
font_identifier "LA"; font_size 14pt#;
font_coding_scheme:="LA text";
mode_setup;
boolean serifs,monospace;
xpt#:=pt#;
monospace:=false;
serifs:=true;
pair auf,ab,an;
auf=(1,7);
ab=(-1,-7);
an=(1,0.8);
ut#:=0.3;
ut=ut#*hppp;
numeric str_d,str_dp,str_dm,str_i;
numeric top,mid,base,bot,an_top,an_mid,an_bot,ang;
numeric kl_oben,gr_oben;
str_d:=1.8ut;         
str_i:=2.8ut;
str_dp:=2.3ut;
str_dm:=1.3ut;
pen va_pen;
va_pen := pencircle scaled str_d;  
pen va_pen_small;
va_pen_small := pencircle scaled str_dm;
pen va_pen_big;
va_pen_big := pencircle scaled str_dp;
top:=8ut; mid:=-3ut; bas:=0ut; bot:=0ut;
kl_oben:=19ut; gr_oben:=34ut;
an_top=16ut; an_mid=11ut; an_bot=6ut;
beginchar(oct"043", 11ut#, 19ut#, 0ut#);    
pickup va_pen;
z1=(0ut,19ut);
z2=(11ut,19ut);
draw z1{dir -70}..tension0.8..{dir 50}z2;
labels(range 1 thru 2);
endchar;
beginchar(oct"044", 12ut#, 19ut#, 0ut#);    
pickup va_pen;
z1=(0ut,19ut);
z2=(11ut,15ut);
draw z1{dir -70}..tension0.8..{dir 50}z2;
labels(range 1 thru 2);
endchar;
beginchar(oct"045", 9ut#, 19ut#, 0ut#);    
pickup va_pen;
z1=(0ut,19ut);
z2=(9ut,14ut);
draw z1{dir -70}..tension0.8..{dir 70}z2;
labels(range 1 thru 2);
endchar;
beginchar(oct"046", 6ut#, 19ut#, 0ut#);    
pickup va_pen;
z1=(0ut,19ut);
z2=(6ut,15ut);
draw z1{dir -70}..tension0.8..{dir 50}z2;
labels(range 1 thru 2);
endchar;
beginchar(133,8ut#,19ut#,0ut);   
pickup va_pen;
z1=(bot,an_bot);
z2=(top-1ut,an_top-1ut);     
draw z1{dir 55}..{dir 50}z2;
labels(range 1 thru 2);
endchar;
beginchar(134,6ut#,19ut#,0ut);   
pickup va_pen;
z1=(bot,an_bot);
z2=(top-2ut,an_top-3ut);          
draw z1{dir 55}..{dir 50}z2;
labels(range 1 thru 2);
endchar;
beginchar(146,7ut#,19ut#,0ut);   
pickup va_pen;
z1=(bot,an_bot);
z2=(top-1ut,an_top+3ut);
draw z1{dir 55}..{dir 70}z2;
labels(range 1 thru 2);
endchar;
beginchar(132,7ut#,19ut#,0ut);   
pickup va_pen;
z1=(bot,an_bot);
z2=(top-1ut,an_top-1ut);
draw z1{dir 55}..{dir 50}z2;
labels(range 1 thru 2);
endchar;
beginchar(127,13ut#,19ut#,0ut);   
pickup va_pen;                      
z1=(bot,an_bot - 6ut);
z2=(top+4.7ut,an_top -3ut);
draw z1{dir 50}..{dir 50}z2;
labels(range 1 thru 2);
endchar;
beginchar(128,8ut#,19ut#,0ut);   
pickup va_pen;                      
z1=(bot,an_bot - 6ut);
z2=(top-0ut,an_top - 10ut);
draw z1..{dir 45}z2;
labels(range 1 thru 2);
endchar;
beginchar(129,14ut#,19ut#,0ut);   
pickup va_pen;                      
z1=(bot,an_bot - 6ut);
z2=(top+6ut,an_top+3ut);
draw z1{dir 50}..{dir 70}z2;
labels(range 1 thru 2);
endchar;
beginchar(130,14ut#,19ut#,0ut);     
pickup va_pen;                        
z1=(bot,an_bot - 6ut);
z2=(top+5ut,an_top - 1ut);
draw z1{dir 45}..{dir 50}z2;
labels(range 1 thru 2);
endchar;
beginchar(154,6ut#,19ut#,0ut);   %
pickup va_pen;                      
z1=(bot,an_bot);
z2=(top-0ut,an_top - 10ut);
draw z1{right}..{dir 45}z2;
labels(range 1 thru 2);
endchar;
beginchar(131,8ut#,19ut#,0ut);     
pickup va_pen;
z1=(bot,an_bot);
z2=(top-1ut,an_top-1ut);
draw z1{dir 55}..{dir 50}z2;
labels(range 1 thru 2);
endchar;
beginchar(143,4ut#,19ut#,0ut);     
pickup va_pen;
z1=(bot    ,an_bot);
z2=(top-4ut,an_top-3ut);
draw z1{dir 55}..{dir 50}z2;
labels(range 1 thru 2);
endchar;
beginchar(144,8ut#,19ut#,0ut);     
pickup va_pen;
z1=(bot,an_bot);
z2=(top-0ut,an_top+3ut);
draw z1{dir 55}..{dir 70}z2;
labels(range 1 thru 2);
endchar;
beginchar(145,7ut#,19ut#,0ut);     
pickup va_pen;
z1=(bot,an_bot);
z2=(top-1ut,an_top-1ut);
draw z1{dir 55}..{dir 50}z2;
labels(range 1 thru 2);
endchar;
beginchar(153,2ut#,19ut#,0ut);   %
pickup va_pen;                      
z1=(bot-12ut,an_bot-6ut);
z2=(top-6ut,an_top - 10ut);
draw z1{right}..{dir 45}z2;
labels(range 1 thru 2);
endchar;
beginchar(139,1ut#,19ut#,0ut);     
pickup va_pen;
z1=(bot-20ut,an_bot-6ut);
z2=(top-8ut,an_top-1ut);
draw z1{right}..{dir 50}z2;
labels(range 1 thru 2);
endchar;
beginchar(140,1ut#,19ut#,0ut);     
pickup va_pen;
z1=(bot-20ut,an_bot-6ut);
z2=(top-6ut,an_top-2ut);
draw z1{right}..{dir 50}z2;
labels(range 1 thru 2);
endchar;
beginchar(141,2ut#,19ut#,0ut);     
pickup va_pen;
z1=(bot-20ut,an_bot-6ut);
z2=(top-6ut,an_top+3ut);
draw z1{right}..{dir 70}z2;
labels(range 1 thru 2);
endchar;
beginchar(142,2ut#,19ut#,0ut);     
pickup va_pen;
z1=(bot-18ut,an_bot-6ut);
z2=(top-6ut,an_top-1ut);
draw z1{right}..{dir 50}z2;
labels(range 1 thru 2);
endchar;
beginchar(148,6ut#,19ut#,0ut);   
pickup va_pen;                      
z1=(bot-5ut,an_bot - 6ut);
z2=(top-2ut,an_top - 10ut);
draw z1{dir 10}..{dir 45}z2;
labels(range 1 thru 2);
endchar;
beginchar(149,11ut#,19ut#,0ut);     
pickup va_pen;                        
z1=(bot-5ut,an_bot - 6ut);
z2=(top+1ut,an_top-3ut);
draw z1{dir 20}..{dir 50}z2;
labels(range 1 thru 2);
endchar;
beginchar(150,9ut#,19ut#,0ut);   
pickup va_pen;                      
z1=(bot-5ut,an_bot - 6ut);
z2=(top+0.7ut,an_top -3ut);
draw z1{dir 20}..{dir 50}z2;
labels(range 1 thru 2);
endchar;
beginchar(151,10ut#,19ut#,0ut);   
pickup va_pen;                      
z1=(bot-5ut,an_bot - 6ut);
z2=(top+2ut,an_top+3ut);
draw z1{dir 20}..{dir 70}z2;
labels(range 1 thru 2);
endchar;
beginchar(152,8ut#,19ut#,0ut);   
pickup va_pen;
z1=(bot-5ut,an_bot - 6ut);
z2=(top+0ut,an_top - 1ut);
draw z1{dir 20}..{dir 50}z2;
labels(range 1 thru 2);
endchar;
beginchar(147,0ut#,19ut#,0ut);   
pickup va_pen;
z1=(bot-20ut,an_bot-6ut);
z2=(top-8ut,an_top-10ut);
draw z1{right}..{dir 45}z2;
labels(range 1 thru 2);
endchar;
beginchar(135,1ut#,19ut#,0ut);     
pickup va_pen;
z1=(bot-23ut,an_bot-6ut);
z2=(top-8ut,an_top-1ut);
draw z1{right}..{dir 50}z2;
labels(range 1 thru 3);
endchar;
beginchar(136,4ut#,19ut#,0ut);     
pickup va_pen;
z1=(bot-23ut,an_bot-6ut);
z2=(top-4ut,an_top+3ut);
draw z1{right}..{dir 70}z2;
labels(range 1 thru 2);
endchar;
beginchar(137,1ut#,19ut#,0ut);     
pickup va_pen;
z1=(bot-20ut,an_bot-6ut);
z2=(top-7ut,an_top-1ut);
draw z1{right}..{dir 50}z2;
labels(range 1 thru 2);
endchar;
beginchar(138,1ut#,19ut#,0ut);     
pickup va_pen;
z1=(bot-23ut,an_bot-6ut);
z2=(top-6ut,an_top-2ut);
draw z1{right}..{dir 50}z2;
labels(range 1 thru 2);
endchar;
beginchar("1",14ut#,26ut#,0);
pickup va_pen;
x18=0 ut; y18=19ut;
x24=11ut; y24=26ut;
x28=5 ut; y28=0 ut;
draw z18--z24
  & z24{ab}--z28;
penlabels(18,24,28);
endchar;
beginchar("2",20ut#,26ut#,0);
pickup va_pen;
x22=17ut*.52; y22=0*.52;
x23=29ut*.52; y23=7 ut*.52;
x24=22ut*.52; y24=50ut*.52;
x28=6ut*.52; y28=7ut*.52;
x29=30ut*.52; y29=46ut*.52;
x34=6 ut*.52; y34=40ut*.52;
x38=22ut*.52; y38=22ut*.52;
x281=0ut*.52; y281=0ut*.52;
draw z34..z24{right}..z29..z38..z28{dir 225}..z281
 & z281{z28-z281}..z22{dir 330}..{an}z23;
penlabels(22,23,24,28,29,34,38,281);
endchar;
beginchar("3",20ut#,50ut#,0);
pickup va_pen;
x1=0ut*.52;   y1=4ut*.52;
x17=24ut*.52; y17=10ut*.52;
x19=12ut*.52; y19=27ut*.52;
x22=13ut*.52; y22=0 ut*.52;
x24=22ut*.52; y24=50ut*.52;
x34=6 ut*.52; y34=40ut*.52;
x38=22ut*.52; y38=22ut*.52;
draw z34..z24{right}..z19{left}
 & z19{right}..z38..z17..z22{left}..z1;
penlabels(1,17,19,22,24,34,38);
endchar;
beginchar("4",18ut#,50ut#,0);
pickup va_pen;
x4=27ut*.52;  y4=16ut*.52;
x22=13ut*.52; y22=0 ut*.52;
x24=15ut*.52; y24=50ut*.52;
x30=19ut*.52; y30=27ut*.52;
x32=16ut*.52; y32=-22ut*.52;
x56=-3ut*.52; y56=16ut*.52;
draw z24{ab}--z56
 & z56{dir 10}--z4;
draw z30--z22;
penlabels(4,22,24,30,32,56);
endchar;
beginchar("5",24ut#,50ut#,0);
pickup va_pen;
x2=8ut*.52;   y2=0ut*.52;
x31=0 ut*.52; y31=7 ut*.52;
x35=12ut*.52; y35=50ut*.52;
x38=28ut*.52; y38=22ut*.52;
x58=35ut*.52; y58=50ut*.52;
x66=2 ut*.52; y66=25ut*.52;
draw z35--z58;
draw z35--z66
 & z66{dir 20}..z38{down}..z2{left}..z31;
penlabels(2,31,35,38,58,66);
endchar;
beginchar("6",17ut#,50ut#,0);
pickup va_pen;
x23=21ut*.52; y23=8ut*.52;
x24=22ut*.52; y24=50ut*.52;
x30=15ut*.52; y30=27ut*.52;
x56=-5ut*.52; y56=15ut*.52;
x57=8ut*.52; y57=0 ut*.52;
draw z56..z30..z23..z57{left}..z56{up}..z24{an};
penlabels(23,24,30,56,57);
endchar;
beginchar("7",21ut#,50ut#,0);
pickup va_pen;
x19=12ut*.52; y19=27ut*.52;
x20=30ut*.52; y20=27ut*.52;
x28=6 ut*.52; y28=0 ut*.52;
x36=25ut*.52; y36=46ut*.52;
x37=32ut*.52; y37=50ut*.52;
x60=3 ut*.52; y60=43ut*.52;
x61=11 ut*.52; y61=50ut*.52;
draw z60..z61{right}..z36{right}..z37
 & z37--z28;
draw z19--z20;
penlabels(19,20,28,36,37,60,61);
endchar;
beginchar("8",22ut#,50ut#,0);
pickup va_pen;
x19=12ut*.52; y19=27ut*.52;
x23=25ut*.52; y23=8ut*.52;
x31=0 ut*.52; y31=9 ut*.52;
x33=11ut*.52; y33=-18ut*.52;
x34=6 ut*.52; y34=40ut*.52;
x38=22ut*.52; y38=22ut*.52;
x57=10ut*.52; y57=0 ut*.52;
x67=27ut*.52; y67=42ut*.52;
x68=18ut*.52; y68=50ut*.52;
x1=27ut*.52; y1=50ut*.52;
x2=31ut*.52; y2=42ut*.52;
x3=39ut*.52; y3=50ut*.52;
draw z31{up}..z19..z67{up}..z68{left}..z34{down}..z38..z23..z57{left}..z31{up};
draw z1{dir -70}..z2{right}..z3;
penlabels(19,23,31,33,34,38,57,67,68,1,2,3);
endchar;
beginchar("9",20ut#,50ut#,0);
pickup va_pen;
x1=0*.52;     y1=4ut*.52;
x3=29ut*.52;  y3=50ut*.52;
x11=32ut*.52; y11=50ut*.52;
x34=6 ut*.52; y34=40ut*.52;
x37=29ut*.52; y37=40ut*.52;
x38=25ut*.52; y38=29ut*.52;
x57=10ut*.52; y57=0 ut*.52;
x68=18ut*.52; y68=50ut*.52;
x69=15ut*.52; y69=22ut*.52;
x70=26ut*.52; y70=40ut*.52;
draw z68{left}..z34..z69{right}..z70..z68{left};
draw z3{ab}..z38{dir 255}..z57{left}..z1;
penlabels(1,3,11,34,37,38,57,68,69,70);
endchar;
beginchar("0",24ut#,50ut#,0);
pickup va_pen;
x1=30ut*.52; y1=50ut*.52;
x2=34ut*.52; y2=42ut*.52;
x3=42ut*.52; y3=50ut*.52;
x17=23ut*.52; y17=10ut*.52;
x20=30ut*.52; y20=41ut*.52;
x24=23ut*.52; y24=50ut*.52;
x46=28ut*.52; y46=22ut*.52;
x56=0ut*.52; y56=15ut*.52;
x57=11ut*.52; y57=0 ut*.52;
x58=36ut*.52; y58=50ut*.52;
draw z20..z24{left}..z56{down}..z57{right}..z17..z46..cycle;
draw z1{dir -70}..z2{right}..z3;
penlabels(1,2,3,17,20,24,46,56,57,58);
endchar;
beginchar("+",30ut#,20ut#, 0ut#);
pickup va_pen;
z1=(10ut,15ut);
z2=(25ut,15ut);
z3=(17ut,7ut);
z4=(18ut,23ut);
draw z1..z2; draw z3..z4;
labels(range 1 thru 4);
endchar;
beginchar("-",30ut#,20ut#, 0ut#);
pickup va_pen;
z1=(10ut,15ut);
z2=(25ut,15ut);
draw z1..z2;
labels(range 1 thru 2);
endchar;
beginchar(oct"025",35ut#,20 ut#, 0ut#);  
pickup va_pen_small;
z1=(10ut,15ut);
z2=(30ut,15ut);
draw z1..z2;
labels(range 1 thru 2);
endchar;
beginchar(oct"026",40ut#,20 ut#, 0ut#);   
pickup va_pen_small;
z1=(10ut,15ut);
z2=(35ut,15ut);
draw z1..z2;
labels(range 1 thru 2);
endchar;
beginchar("*",30 ut#,20 ut#, 0ut#);
pickup va_pen;
z9=(15ut,15ut);
z1=z9+(0ut,-8ut);
z2=z9+(8ut,0ut);
z3=z9+(0ut,8ut);
z4=z9+(-8ut,0ut);
z5=z9+(-5ut,-5ut);
z6=z9+(5ut,-5ut);
z7=z9+(5ut,5ut);
z8=z9+(-5ut,5ut);
draw z1..z3; draw z2..z4; draw z5..z7; draw z6..z8;
labels(range 1 thru 9);
endchar;
beginchar(":",15 ut#,20 ut#, 0ut#);
pickup pencircle scaled  str_i;
z1=(10ut,15ut);
z2=(9ut,5ut);
drawdot z1; drawdot z2;
labels(range 1 thru 2);
endchar;
beginchar("<",30 ut#,20 ut#,0 ut#);
pickup va_pen;
z1=(28ut,23ut);
z2=(10ut,15ut);
z3=(26ut,7ut);
draw z1..z2 & z2..z3;
labels(range 1 thru 3);
endchar;
beginchar(">", 30ut#, 20ut#, 0ut#);
pickup va_pen;
z1=(13ut,23ut);
z2=(30ut,15ut);
z3=(12ut,7ut);
draw z1..z2 & z2..z3;
labels(range 1 thru 3);
endchar;
beginchar("=", 30ut#,20 ut#, 0ut#);
pickup va_pen;
z1=(10ut,15ut);
z2=(25ut,15ut);
z3=(x1,y1-8ut);
z4=(x2,y2-8ut);
draw z1..z2  ; draw z3..z4;
labels(range 1 thru 4);
endchar;
beginchar("(", 13ut#, 34ut#, 0ut#);
pickup va_pen;
z1=(11ut,34ut);
z2=(3ut,-2ut);
draw z1{dir -115}..{dir -80}z2;
labels(range 1 thru 2);
endchar;
beginchar(")",13 ut#,34 ut#, 0ut#);
pickup va_pen;
z1=(11ut,34ut);
z2=(3ut,-2ut);
draw z1{dir -80}..{dir -115}z2;
labels(range 1 thru 2);
endchar;
beginchar("/",15 ut#,34 ut#, 0ut#);
pickup va_pen;
z1=(12ut,34ut);
z2=(3ut,-2ut);
draw z1..z2;
labels(range 1 thru 2);
endchar;
beginchar(oct"022",18 ut#,20 ut#, 7ut#);   
pickup va_pen;
z1=(5ut,3ut);
z2=(3ut,-7ut);
z3=(x1+8ut,y1);
z4=(x2+8ut,y2);
draw z1{dir -90}..z2;
draw z3{dir -90}..z4;
labels(range 1 thru 4);
endchar;
beginchar(oct"021", 18ut#, 36ut#, ut#);   
pickup va_pen;
z1=(5ut,36ut);
z2=(3ut,26ut);
z3=(x1+8ut,y1);
z4=(x2+8ut,y2);
draw z1{dir -90}..z2;  draw z3{dir -90}..z4;
labels(range 1 thru 4);
endchar;
beginchar(oct"20", 18ut#, 36ut#, ut#); 
pickup va_pen;
z1=(5ut,36ut);
z2=(3ut,26ut);
z3=(x1+8ut,y1);
z4=(x2+8ut,y2);
draw z1{dir -90}..z2;  draw z3{dir -90}..z4;
labels(range 1 thru 4);
endchar;
beginchar(oct"042", 18ut#, 36ut#, ut#); 
pickup va_pen;
z1=(5ut,36ut);
z2=(7ut,26ut);
z3=(x1+8ut,y1);
z4=(x2+8ut,y2);
draw z1{dir -90}..z2;  draw z3{dir -90}..z4;
labels(range 1 thru 4);
endchar;
beginchar(oct"140", 7ut#, 36ut#, ut#); 
pickup va_pen;
z1=(5ut,36ut);
z2=(7ut,26ut);
draw z1{dir -90}..z2;
labels(range 1 thru 2);
endchar;
beginchar(oct"000", 7ut#, 36ut#, ut#); 
pickup va_pen;
z1=(5ut,36ut);
z2=(7ut,26ut);
draw z1{dir -90}..z2;
labels(range 1 thru 2);
endchar;
beginchar(oct"004",18ut#,15ut#,0ut#);   
pickup va_pen_big;
z1=(-3ut,22ut);
z2=z1+(9ut,0ut);
z3=z1+(-1ut,-1ut);
z4=z2+(-1ut,-1ut);
draw z1..z3; draw z2..z4;
labels(range 1 thru 4);
endchar;
beginchar("?", 26ut#,34 ut#,0 ut#);
pickup va_pen;
z1=(10ut,30ut);
z2=(18ut,34ut);
z3=(23ut,30ut);
z4=(15ut,20ut);
z5=(8ut,10ut);
z6=(14ut,6ut);
z7=(21ut,10ut);
z8=(15ut,1ut);
draw z1..z2{right}..z3..tension1.4..z4..tension1.4..z5{down}..z6{right}..z7;
pickup pencircle scaled  str_i;
drawdot z8;
labels(range 1 thru 8);
endchar;
beginchar("!",16 ut#,34 ut#, 0ut#);
pickup va_pen;
z1=(16ut,32ut);
z2=(9ut,9ut);
z3=(8ut,2ut);
draw z1..z2;
pickup pencircle scaled  str_i;
drawdot z3;
labels(range 1 thru 3);
endchar;
beginchar(";",7 ut#,20 ut#,10 ut#);
pickup va_pen;
z1=(6ut,5ut);
z2=(6ut,1ut);
z3=(2ut,-7ut);
draw z2{dir -90}..z3{dir -110};
pickup pencircle scaled  str_i;
drawdot z1;
labels(range 1 thru 3);
endchar;
beginchar(".",6 ut#,10 ut#,0 ut#);
pickup pencircle  scaled str_i;
z1=(5ut,1ut);
drawdot z1;
labels(range 1 thru 2);
endchar;
beginchar(",", 7ut#, 10ut#,7 ut#);
pickup va_pen;
z1=(6ut,3ut);
z2=(2ut,-7ut);
drawdot z1;
draw z1+(1ut,0ut){dir -80}..z2{dir -125};
labels(range 1 thru 2);
endchar;
beginchar(oct"47", 7ut#, 40ut#,0 ut#);  
pickup va_pen;
z1=(6ut,35ut);
z2=(2ut,27ut);
drawdot z1;
draw z1+(1ut,0ut){dir -80}..z2{dir -125};
labels(range 1 thru 8);
endchar;
beginchar(oct"001", 7ut#, 40ut#,0 ut#);  
pickup va_pen;
z1=(6ut,35ut);
z2=(2ut,27ut);
drawdot z1;
draw z1+(1ut,0ut){dir -80}..z2{dir -125};
labels(range 1 thru 8);
endchar;
beginchar("a",18 ut#,19 ut#,0);       
pickup va_pen;
z1=(6ut,19ut);
z2=(-4.5ut,8ut);
z3=(0ut,0ut);
z4=(10ut,13ut);
z5=(13ut,19ut);
z6=(8ut,5ut);
z7=(10ut,0ut);
z8=(18ut,6ut);
draw z1{dir 190}..{dir 250}z2{dir 250}..z3{right}..{dir 70}z4{dir 70}..cycle;
draw z5--z6{z6-z5}..z7{dir 10}..{dir 55}z8;
labels(range 1 thru 8);
endchar;
beginchar(oct"344",18 ut#,30ut#,0);       
pickup va_pen;
z1=(6ut,19ut);
z2=(-4.5ut,8ut);
z3=(0ut,0ut);
z4=(10ut,13ut);
z5=(13ut,19ut);
z6=(8ut,5ut);
z7=(10ut,0ut);
z8=(18ut,6ut);
z9=(4ut,24ut);
z10=(9ut,24ut);
z11=(6.0ut,30ut);
z12=(11.0ut,30ut);
draw z1{dir 190}..{dir 250}z2{dir 250}..z3{right}..{dir 70}z4{dir 60}..cycle;
draw z5--z6{z6-z5}..z7{dir 10}..{dir 55}z8;
draw z9--z11;
draw z10--z12;
labels(range 1 thru 12);
endchar;
beginchar("b",12 ut#,34 ut#,0);
pickup va_pen;
z1=(0ut,13ut);
z2=(12ut,31ut);
z3=(9ut,34ut);
z4=(4ut,0ut);
z5=(12ut,19ut);
z6=(25ut,15ut);
z7=(3ut,23ut);
draw z1{dir 50}..z2{up}..z3{dir 220}..z7..{dir 260}z1{dir 260}
     ..z4{dir 10}..z5{dir 110};
labels(range 1 thru 7);
endchar;
beginchar("c",10 ut#,19 ut#,0);
pickup va_pen;
z1=(10ut,18ut);
z2=(6ut,19ut);
z3=(-4.5ut,8ut);
z4=(2ut,0ut);
z5=(10ut,6ut);
draw z1..z2{dir 190}..{dir 250}z3{dir 250}..z4{dir 10}..{dir 55}z5;
labels(range 1 thru 5);
endchar;
beginchar("d",18 ut#,34 ut#,0);
pickup va_pen;
z1=(6ut,19ut);
z2=(-4.5ut,8ut);
z3=(0ut,0ut);
z4=(11ut,13ut);
z5=(18ut,34ut);
z6=(9ut,5ut);
z7=(10ut,0ut);
z8=(25ut,6ut);
z9=(18ut,6ut);
draw z1{dir 190}..{dir 250}z2{dir 250}..z3{right}..z4{dir 70}..cycle;
draw z5--z6{z6-z5}..z7{dir 10}..z9{dir 55};
labels(range 1 thru 9);
endchar;
beginchar("e",13 ut#,19 ut#,0);
pickup va_pen;
z1=(6ut,19ut);
z2=(4ut,0ut);
z3=(13ut,6ut);
z4=(0ut,6ut);
z5=(7ut,12ut);
z6=(0ut,7ut);
draw z4{dir 45}..z5
     ..z1{dir 210}..{dir 250}z6{dir 250}..z2{dir 10}..z3{dir 55};
labels(range 1 thru 6);
endchar;
beginchar(oct"037",21 ut#,19 ut#,0);  
pickup va_pen;
z1=(15ut,19ut);
z2=(7ut,12ut);
z3=(12ut,0ut);
z4=(21ut,6ut);
z5=(14ut,10ut);
z6=(18ut,14ut);
z7=(8ut,8ut);
z8=(0ut,19ut);
draw z8{dir -70}..z2..z5..z6..z1{left}
     ..{dir 260}z7{dir 260}..z3{right}..z4{dir 55};
labels(range 1 thru 8);
endchar;
beginchar("f",9 ut#,34 ut#,15ut#);
pickup va_pen;
z1=(0ut,13ut);
z2=(12ut,32ut);
z3=(9ut,34ut);
z4=(4ut,26ut);
z5=(-7ut,-15ut);
z6=(-4ut,5ut);
z7=(3ut,2ut);
z8=(9ut,6ut);
draw z1{dir 50}..z2{up}..z3{dir 220}...z4{z5-z4}--z5;
draw z6..z7...z8{dir 55};
labels(range 1 thru 8);
endchar;
beginchar("g",17 ut#,19 ut#,15ut#);
pickup va_pen;
z1=(11ut,13ut);
z2=(0ut,0ut);
z3=(-4.5ut,8ut);
z4=(6ut,19ut);
z10=(13ut,19ut);
z11=(5ut,-6ut);
z12=(0ut,-15ut);
z13=(-3ut,-11ut);
z14=(1ut,-5ut);
z15=(17ut,6ut);
draw z1{dir 250}..z2{left}..{dir 70}z3{dir 70}..{dir 10}z4..cycle;
draw z10--z11{z11-z10}..z12{left}..z13{up}..z14...z15{dir 55};  
labels(range 1 thru 15);
endchar;
beginchar("h",20 ut#,34 ut#,0);
pickup va_pen;
z1=(0ut,13ut);
z2=(12ut,31ut);
z3=(9ut,34ut);
z4=(3ut,23ut);
z5=(-4ut,0ut);
z6=(11ut,17ut);
z7=(13ut,15ut);
z8=(10ut,3ut);
z9=(12ut,0ut);
z10=(20ut,6ut);
draw z1{dir 50}..z2{up}..z3{dir 220}..z4..z1--z5
  & z5{dir 70}..tension1.3..z6{right}..z7{down}..z8{z8-z7}..z9{dir 10}..z10{dir 55};
labels(range 1 thru 10);
endchar;
beginchar("i",5 ut#,19 ut#,0);
pickup va_pen;
z1=(0ut,19ut);
z2=(-5ut,4ut);
z3=(-3ut,0ut);
z4=(5ut,6ut);
z5=(2.5ut,28ut);
draw z1--z2{z2-z1}..z3{dir 10}..z4{dir 55};
pickup pencircle scaled  str_i;
drawdot z5;
labels(range 1 thru 5);
endchar;
beginchar("j",7 ut#,19 ut#,15ut#);
pickup va_pen;
z1=(2.5ut,28ut);
z2=(0ut,19ut);
z3=(-4ut,0ut);
z4=(-12ut,-15ut);
z5=(7ut,6ut);
draw z2--z3{z3-z2}...z4{left}..z3{z5-z3}...{dir 55}z5;
pickup pencircle scaled  str_i;
drawdot z1;
labels(range 1 thru 5);
endchar;
beginchar("k", 19 ut#, 34ut#,0);
pickup va_pen;
z1=(0ut,13ut);
z2=(12ut,31ut);
z3=(9ut,34ut);
z4=(3ut,23ut);
z5=(-4ut,0ut);
z6=(11ut,17ut);
z7=(13ut,15ut);
z8=(3ut,9ut);
z9=(11ut,0ut);
z10=(19ut,6ut);
z11=(16ut,2ut);
draw z1{dir 50}..z2{up}..z3{dir 220}...z4{z5-z4}--z5
  & z5{dir 70}..tension1.3..z6{right}..z7{down}..z8
  & z8{right}..z9{dir 10}..z10{dir 55};
labels(range 1 thru 11);
endchar;
beginchar("l",9 ut#,34 ut#,0);
pickup va_pen;
z1=(0ut,13ut);
z2=(12ut,31ut);
z3=(9ut,34ut);
z4=(1ut,0ut);
z5=(9ut,6ut);
z6=(7ut,2ut);
z7=(3ut,23ut);
draw z1{dir 50}..z2{up}..z3{dir 220}..z7{z1-z7}..z1{z1-z7}    
     ..{right}z4{dir 10}..z5{dir 55};
labels(range 1 thru 7);
endchar;
beginchar("m",36 ut#,19 ut#,0);
pickup va_pen;
z1=(5ut,15ut);
z2=(0ut,0ut);
z3=(16.5ut,19ut);
z4=(18.5ut,16ut);
z5=(13.5ut,0ut);
z6=(29ut,19ut);
z7=(31ut,16ut);
z8=(27ut,3ut);
z9=(28ut,0ut);
z10=(36ut,6ut);
z11=(0ut,15ut);
z12=(4ut,19ut);
draw z11{dir 50}..{curl4}z12{curl6}..{z2-z1}z1--z2;
draw z2{z1-z2}..tension1.3..{curl2}z3{curl4}..{z5-z4}z4--z5;
draw z5{z4-z5}..tension1.3..{curl2}z6{curl4}..{z8-z7}z7
      --z8{z8-z7}..z9{dir 10}..z10{dir 55};
labels(range 1 thru 12);
endchar;
beginchar("n",24 ut#,19 ut#,0);
pickup va_pen;
z1=(5ut,15ut);
z2=(0ut,0ut);
z3=(17ut,19ut);
z4=(19.3ut,16ut);
z5=(15ut,3ut);
z6=(16ut,0ut);
z7=(24ut,6ut);
z8=(0ut,15ut);
z9=(4ut,19ut);
draw z8{dir 50}..{curl4}z9{curl6}..{z2-z1}z1--z2;
draw z2{z1-z2}..tension1.3..{curl2}z3{curl4}
     ..{z5-z4}z4--z5{z5-z4}..z6{dir 10}..z7{dir 55};
labels(range 1 thru 9);
endchar;
beginchar("o",11 ut#,19 ut#,0);
pickup va_pen;
z1=(11ut,16ut);
z2=(0ut,0ut);
z3=(-4.5ut,8ut);
z4=(6ut,19ut);
draw z1{dir -75}..z2{left}..{dir 70}z3{dir 70}
     ..{dir 10}z4..cycle;
labels(range 1 thru 4);
endchar;
beginchar(oct"366",11 ut#,30 ut#,0);       
pickup va_pen;
z1=(11ut,16ut);
z2=(0ut,0ut);
z3=(-4.5ut,8ut);
z4=(6ut,19ut);
z5=(6ut,24ut);
z6=(11ut,24ut);
z7=(8.0ut,30ut);
z8=(13.0ut,30ut);
draw z1{dir -75}..z2{left}..{dir 70}z3{dir 70}
     ..{dir 10}z4..cycle;
draw z5--z7;
draw z6--z8;
labels(range 1 thru 8);
endchar;
beginchar("p",18 ut#,19 ut#,15ut#);
pickup va_pen;
z1=(0ut,19ut);
z2=(-10ut,-15ut);
z3=(-4.5ut,4ut);
z4=(11.5ut,19ut);
z5=(13.0ut,15ut);
z6=(9ut,4ut);
z7=(10ut,0ut);
z8=(17ut,1ut);
z9=(18ut,6ut);
z10=(6ut,16ut);
draw z1--z2;
draw z3{dir 50}..z10..z4{right}..z5{z6-z5}--z6{z6-z5}..z7{dir 10}..z9{dir 55};
labels(range 1 thru 10);
endchar;
beginchar("q",23 ut#,19 ut#,15ut#);
pickup va_pen;
z1=(9ut,13ut);
z2=(0ut,0ut);
z3=(-4.5ut,8ut);
z4=(6ut,19ut);
z5=(11ut,19ut);
z6=(1ut,-15ut);
z7=(6ut,0ut);
z8=(23ut,19ut);
draw z1{dir 250}..z2{left}..{dir 70}z3{dir 70}..{dir 10}z4..cycle;
draw z5--z6;
draw z7{dir 40}..z8{dir 70};
labels(range 1 thru 8);
endchar;
beginchar("r", 10ut#,19 ut#,0);  
pickup va_pen;
z1=(5ut,15ut);
z2=(4ut,11ut);    
z3=(10ut,19ut);
z4=(0ut,0ut);
z5=(0ut,15ut);
z6=(4ut,19ut);
draw z5{dir 50}..{curl4}z6{curl6}..{z4-z1}z1--z4;
draw z2..z3;
labels(range 1 thru 6);
endchar;
beginchar("s",4 ut#,19 ut#,0);
pickup va_pen;
z1=(-10ut,7ut);
z2=(0ut,19ut);
z3=(3ut,6ut);
z4=(-2ut,0ut);
z5=(-6ut,2ut);
draw z2{dir 260}..tension0.8..z3{down}..z4{left}..z5;
labels(range 1 thru 5);
endchar;
beginchar("t",7 ut#,34 ut#,0);
pickup va_pen;
z1=(4ut,27ut);
z2=(-3.5ut,3ut);
z3=(-1ut,0ut);
z4=(7ut,6ut);
z5=(-4ut,18ut);
z6=(7ut,18ut);
draw z1..z2{z2-z1}..z3{dir 10}..z4{dir 55};
draw z5--z6;
labels(range 1 thru 6);
endchar;
beginchar("u",19 ut#,19 ut#,0);
pickup va_pen;
z1=(0ut,19ut);
z2=(-5ut,3ut);
z3=(0ut,0ut);
z4=(14ut,19ut);
z5=(9ut,3ut);
z6=(11ut,0ut);
z7=(19ut,6ut);
z8=(17ut,2ut);
draw z1--z2{z2-z1}..z3{dir 20}..tension1.3..z4{z4-z5}
  & z4--z5{z5-z4}..z6{dir 10}..{dir 55}z7;
labels(range 1 thru 8);
endchar;
beginchar(oct"374",19 ut#,30 ut#,0);       
pickup va_pen;
z1=(0ut,19ut);
z2=(-5ut,3ut);
z3=(0ut,0ut);
z4=(14ut,19ut);
z5=(9ut,3ut);
z6=(11ut,0ut);
z7=(19ut,6ut);
z8=(6ut,24ut);
z9=(11ut,24ut);
z10=(7.5ut,30ut);
z11=(12.5ut,30ut);
draw z1--z2{z2-z1}..z3{dir 20}..tension1.3..z4{z4-z5}
  & z4--z5{z5-z4}..z6{dir 10}..{dir 55}z7;
draw z8--z10;
draw z9--z11;
labels(range 1 thru 12);
endchar;
beginchar("v",17 ut#,19 ut#,0);
pickup va_pen;
z1=(7ut,15ut);
z2=(3ut,5ut);
z3=(9ut,0ut);
z4=(17ut,19ut);
z5=(30ut,15ut);
z6=(0ut,15ut);
z7=(4ut,19ut);
draw z6{dir 50}..z7{right}..{z2-z1}z1..z2{z2-z1}..z3{dir 20}..z4;
labels(range 1 thru 7);
endchar;
beginchar("w",29 ut#,19 ut#,0);
pickup va_pen;
z1=(8ut,15ut);
z2=(5ut,6ut);
z3=(7ut,0ut);
z4=(16ut,7ut);
z5=(19ut,19ut);
z6=(21ut,0ut);
z7=(29ut,19ut);
z8=(42ut,15ut);
z9=(0ut,15ut);
z10=(4ut,19ut);
draw z9{dir 50}..z10{right}..{z2-z1}z1--z2{z2-z1}..z3{right}..z4--z5
  & z5--z4..z6{right}..z7{dir 100};
labels(range 1 thru 8);
endchar;
beginchar("x",19 ut#,19 ut#,0);
pickup va_pen;
z1=(0ut,15ut);
z2=(5ut,19ut);
z3=(9.5ut,13ut);
z4=(7.5ut,6ut);
z5=(0ut,0ut);
z6=(1ut,6ut);
z7=(13ut,10ut);
z8=(16ut,19ut);
z9=(11ut,0ut);
z10=(19ut,6ut);
draw z1{dir 50}..z2{right}..{z4-z3}z3--z4{z4-z3}..z5{left}..{z7-z6}z6;
draw z6--z7{z7-z6}..z8{left}..{z4-z3}z3;
draw z4{z4-z3}..z9{dir 10}..z10{dir 55};
labels(range 1 thru 10);
endchar;
beginchar("y",19 ut#,19 ut#,15ut#);
pickup va_pen;
z1=(0ut,19ut);
z2=(-4.4ut,5ut);
z3=(-2ut,0ut);
z4=(12ut,19ut);
z5=(2ut,-13ut);
z6=(-1ut,-15ut);
z7=(-4ut,-12ut);
z8=(1ut,-5ut);
z9=(19ut,6ut);
draw z1--z2{z2-z1}..z3{dir 20}..z4{z4-z5}
  & z4--z5..z6{left}..z7{up}..z8{z9-z8}..z9{dir 55};
labels(range 1 thru 9);
endchar;
beginchar("z",18 ut#,19 ut#,0ut#);
pickup va_pen;
z1=(0ut,15ut);
z2=(3ut,19ut);
z3=(11ut,17ut);
z4=(16ut,19ut);
z5=(-3ut,0ut);
z6=(0ut,3ut);
z7=(10ut,0ut);
z8=(18ut,6ut);
draw z1{dir 50}..z2{right}..z3..z4{dir 55}
  & z4--z5
  & z5{dir 30}..z6{right}..z7{dir 10}..z8{dir 55};
labels(range 1 thru 8);
endchar;
beginchar("A",35ut#,34ut#,0);
pickup va_pen;
z1=(0ut,2ut);
z2=(12ut,6ut);
z3=(35ut,34ut);
z4=(26ut,0ut);
z5=(12ut,14ut);
z6=(27ut,6ut);
z7=(26ut,28ut);
z8=(22ut,16ut);
z9=(22ut,3ut);
z10=(35ut,6ut);
z11=(3ut,0ut);
draw z1{down}..z11..z2..z7..z3
  & z3--z4;
  draw z6{dir 100}..z8..z5..z9..{dir 55}z10;
labels(range 1 thru 11);
endchar;
beginchar(oct"304",35ut#,44ut#,0);   
pickup va_pen;
z1=(0ut,2ut);
z2=(12ut,6ut);
z3=(35ut,34ut);
z4=(26ut,0ut);
z5=(12ut,14ut);
z6=(27ut,6ut);
z7=(26ut,28ut);
z8=(22ut,16ut);
z9=(22ut,3ut);
z10=(35ut,6ut);
z11=(3ut,0ut);
z12=(32ut,38ut);
z13=(33.5ut,44ut);
z14=(37ut,38ut);
z15=(38.5ut,44ut);
draw z1{down}..z11..z2..z7..z3
  & z3--z4;
  draw z6{dir 100}..z8..z5..z9..{dir 55}z10;
  draw z12..z13;
  draw z14..z15;
labels(range 1 thru 15);
endchar;
beginchar("B",45ut#,34ut#,0);
pickup va_pen;
z1=(8ut,28ut);
z2=(24ut,32ut);
z3=(31ut,28ut);
z4=(31ut,8ut);
z5=(0ut,2ut);
z6=(21ut,34ut);
z7=(23ut,0ut);
z8=(16ut,2ut);
z9=(12ut,6ut);
z10=(6ut,0ut);
z11=(21ut,18ut);
z12=(20ut,30ut);
draw z1..z6{right}..z3{down}..z11{left}
  & z11{right}..z4{down}..z7{left}..z8;
  draw z12--z9..z10..z5;
labels(range 1 thru 12);
endchar;
beginchar("C",19ut#,34ut#,0ut#);
pickup va_pen;
z1=(24ut,26ut);
z2=(20ut,34ut);
z3=(1ut,10ut);
z4=(7ut,0ut);
z5=(19ut,6ut);
z6=(14ut,2ut);
z7=(12ut,20ut);
z8=(0ut,24ut);
z9=(6ut,23ut);
draw z8..z7..z1..z2{left}..z9..z3{z3-z9}..z4{right}..z6..{dir 55}z5;
labels(range 1 thru 9);
endchar;
beginchar("D",39ut#,34ut#,0ut#);
pickup va_pen;
z1=(3ut,30ut);
z2=(21ut,32ut);
z3=(26ut,20ut);
z4=(18ut,0ut);
z5=(15ut,30ut);
z6=(8ut,6ut);
z7=(3ut,0ut);
z8=(0ut,3ut);
z9=(14ut,34ut);
draw z5--z6{z7-z6}..z7..z8..{right}z6{z4-z6}..z4{dir 30}..z3{up}..z2..z9..z1;
labels(range 1 thru 9);
endchar;
beginchar("E",20ut#,34ut#,0ut#);
pickup va_pen;
z1=(23ut,31ut);
z2=(6ut,25ut);
z3=(15ut,18ut);
z4=(0ut,7ut);
z5=(9ut,0ut);
z6=(20ut,6ut);
draw z1..z2{down}..z3
  & z3{left}..z4..z5{right}..{dir 55}z6;   
labels(range 1 thru 6);
endchar;
beginchar("F",32ut#,34ut#,0ut#);
pickup va_pen;
z1=(14ut,34ut);
z2=(12ut,8ut);
z3=(2ut,22ut);
z4=(32ut,34ut);
z5=(9ut,18ut);
z6=(21ut,18ut);
z7=(18ut,31ut);
z8=(6ut,0ut);
z9=(0ut,2ut);
draw z3{up}..z1{right}..z4;
draw z7--z2{z2-z7}..z8..z9;
draw z5--z6;
labels(range 1 thru 9);
endchar;
beginchar("G",25ut#,34ut#,15ut#);
pickup va_pen;
z1=(22ut,26ut);
z2=(19ut,34ut);
z3=(0ut,8ut);
z4=(5ut,0ut);
z5=(20ut,16ut);
z6=(15ut,0ut);
z7=(3ut,-15ut);
z8=(1ut,-11ut);
z9=(8ut,-4ut);
z10=(25ut,6ut);
z11=(0ut,23ut);
draw z11{dir -30}..z1..z2{left}..z3{down}..z4{right}..z5{z5-z6}
  & z5--z6{z6-z5}..z7{left}..z8..z9{z6-z9}..{dir 55}z10;
labels(range 1 thru 11);
endchar;
beginchar("H",30ut#,34ut#,0ut#);
pickup va_pen;
z1=(1ut,29ut);
z2=(5ut,34ut);
z3=(10ut,32ut);
z4=(15ut,34ut);
z5=(10ut,12ut);
z6=(4ut,0ut);
z7=(-1ut,4ut);
z8=(7ut,14ut);
z9=(27ut,24ut);
z10=(31ut,31ut);
z11=(27ut,34ut);
z12=(23ut,31ut);
z13=(17ut,6ut);
z14=(22ut,0ut);
z15=(30ut,6ut);
draw z1..z2{right}..z3..z4;
draw z4--z5{dir 260}..z6..z7..z8{dir 35}--z9{dir 35}..z10..z11..z12--z13;
draw z13{dir 260}..z14{dir 10}..{dir 55}z15;
labels(range 1 thru 15);
endchar;
beginchar("I",29ut#,34ut#,0ut#);
pickup va_pen;
z1=(9ut,34ut);
z2=(25ut,34ut);
z3=(17ut,8ut);
z4=(9ut,0ut);
z5=(1ut,3ut);
z6=(20ut,32ut);
z7=(6ut,31ut);
draw z7{dir 80}..z1{right}..z6..z2{right}
  & z2--z3{z3-z2}..z4{left}..z5;
labels(range 1 thru 7);
endchar;
beginchar("J",22ut#,34ut#,-15ut#);
pickup va_pen;
z1=(9ut,34ut);
z2=(25ut,34ut);
z3=(13ut,0ut);
z4=(4ut,-15ut);
z5=(0ut,-11ut);
z6=(7ut,-4ut);
z7=(22ut,6ut);
z8=(6ut,31ut);
z9=(20ut,32ut);
draw z8..z1{right}..z9..z2{right}
 & z2--z3{z3-z2}...z4{left}..z5{dir 70}..z6{z3-z6}..{dir 55}z7;
labels(range 1 thru 9);
endchar;
beginchar("K",34ut#,34ut#,0ut#);
pickup va_pen;
z1=(5ut,30ut);
z2=(9ut,34ut);
z3=(14ut,32ut);
z4=(19ut,34ut);
z5=(12ut,8ut);
z6=(6ut,0ut);
z7=(0ut,2ut);
z8=(16ut,19ut);
z9=(28ut,30ut);
z10=(35ut,34ut);
z11=(35ut,37ut);
z12=(20ut,6ut);
z13=(26ut,0ut);
z14=(34ut,6ut);
draw z1..z2{right}..z3..z4;
draw z4--z5{z5-z4}..z6..z7;
draw z8{dir 20}..z9{z9-z8}..z10{dir 20};
draw z8--z12{z12-z8}..z13{dir 10}..{dir 55}z14;
labels(range 1 thru 14);
endchar;
beginchar("L",27ut#,34ut#,0ut#);
pickup va_pen;
z1=(6ut,23ut);
z2=(27ut,26ut);
z3=(27ut,6ut);
z4=(20ut,0ut);
z5=(17ut,29ut);
z6=(9ut,7ut);
z7=(4ut,0ut);
z8=(0ut,3ut);
z9=(22ut,34ut);
draw z1{dir 330}..z2..z9..{z6-z5}z5--z6{z6-z5}..z7..z8{up}
     ..z6{right}..z4{dir 10}..{dir 55}z3;
labels(range 1 thru 9);
endchar;
beginchar("M",53ut#,34ut#,0ut#);
pickup va_pen;
z1=(0ut,1ut);
z2=(10ut,6ut);
z3=(31ut,34ut);
z4=(24ut,0ut);
z5=(48ut,34ut);
z6=(40ut,4ut);
z7=(45ut,0ut);
z8=(53ut,6ut);
z9=(49ut,2ut);
z10=(4ut,0ut);
z11=(24ut,26ut);
draw z1..z10..{z11-z2}z2..z11{z11-z2}..z3
 & z3--z4
 & z4--z5
 &z5--z6{z6-z5}..z7{dir 10}..{dir 55}z8;
labels(range 1 thru 11 );
endchar;
beginchar("N",45ut#,34ut#,0ut#);
pickup va_pen;
z1=(0ut,1ut);
z2=(9ut,5ut);
z3=(28ut,34ut);
z4=(23ut,0ut);
z5=(43ut,31ut);
z6=(48ut,34ut);
z7=(4ut,0ut);
z8=(23ut,28ut);
draw z1..z7..{z8-z2}z2..z8{z3-z2}..z3
 & z3--z4
 & z4--z5{z5-z4}..z6{dir 20};
labels(range 1 thru 9);
endchar;
beginchar("O",34ut#,34ut#,0ut#);
pickup va_pen;
z1=(24ut,23ut);
z2=(9ut,0ut);
z3=(1ut,16ut);
z4=(18ut,34ut);
z5=(23ut,34ut);
z6=(28ut,29ut);
z7=(33ut,34ut);
draw z1{down}..z2{left}..z3{dir 80}..z4{right}..cycle;
draw z5.{down}..z6..{dir 50}z7;
labels(range 1 thru 7);
endchar;
beginchar(oct"326",34ut#,46ut#,0ut#);   
pickup va_pen;
z1=(24ut,23ut);
z2=(9ut,0ut);
z3=(1ut,16ut);
z4=(18ut,34ut);
z5=(23ut,34ut);
z6=(28ut,29ut);
z7=(33ut,34ut);
z8=(16.5ut,44ut);
z9=(21.5ut,44ut);
z10=(15ut,38ut);
z11=(20ut,38ut);
draw z1{down}..z2{left}..z3{dir 80}..z4{right}..cycle;
draw z5.{down}..z6..{dir 50}z7;
draw z8..z10;
draw z9..z11;
labels(range 1 thru 11);
endchar;
beginchar("P",32ut#,34ut#,0ut#);
pickup va_pen;
z1=(0ut,2ut);
z2=(6ut,0ut);
z3=(12ut,8ut);
z4=(18ut,30ut);
z5=(4ut,24ut);
z6=(19ut,34ut);
z7=(32ut,26ut);
z8=(18ut,19ut);
draw z1..z2{right}..z3{z4-z3}...z4;
draw z5{dir 60}..z6{right}..{down}z7{down}..{dir 180}z8;
labels(range 1 thru 8);
endchar;
beginchar("Q",37ut#,34ut#,0ut#);
pickup va_pen;
z1=(24ut,23ut);
z2=(9ut,0ut);
z3=(1ut,16ut);
z4=(18ut,34ut);
z5=(10ut,4ut);
z6=(17ut,6ut);
z7=(23ut,0ut);
z8=(37ut,19ut);
z9=(23ut,34ut);
z10=(26ut,29ut);
z11=(32ut,34ut);
draw z1{down}..z2{left}..z3{dir 80}..z4{right}..cycle;
draw z5..z6{right}..z7{dir 40}..{dir 70}z8;
draw z9.{down}..z10..{dir 50}z11;
labels(range 1 thru 11);
endchar;
beginchar("R",36ut#,34ut#,0ut#);
pickup va_pen;
z1=(0ut,2ut);
z2=(6ut,0ut);
z3=(12ut,8ut);
z4=(18ut,30ut);
z5=(4ut,25ut);
z6=(17ut,34ut);
z7=(32ut,26ut);
z8=(18.5ut,19ut);
z9=(21.5ut,18ut);
z10=(26ut,4ut);
z11=(30ut,0ut);
z12=(36ut,6ut);
draw z1..z2{right}..z3{z4-z3}...z4;
draw z5..z6{right}..{dir 270}z7{dir 270}..{dir 185}z8;
draw z9{z12-z9}--z10{z10-z9}..z11{dir 10}..{dir 55}z12;
labels(range 1 thru 12);
endchar;
beginchar("S",27ut#,34ut#,0ut#);
pickup va_pen;
z1=(7ut,24ut);
z2=(22ut,22ut);
z3=(27ut,34ut);
z4=(17ut,27ut);
z5=(15ut,19ut);
z6=(12ut,8ut);
z7=(6ut,0ut);
z8=(0ut,2ut);
draw z1{dir -40}..z2..z3{left}..z4{z6-z4}..z5{z6-z4}..z6{z6-z4}..z7{left}..z8;
labels(range 1 thru 8);
endchar;
beginchar("T",31ut#,34ut#,0ut#);
pickup va_pen;
z1=(15ut,34ut);
z2=(12ut,8ut);
z3=(2ut,24ut);
z4=(31ut,34ut);
z5=(17.5ut,30ut);
z6=(6ut,0ut);
z7=(0ut,2ut);
draw z3{up}..z1{right}..z4;
draw z5--z2{z2-z5}..z6..z7;
labels(range 1 thru 7);
endchar;
beginchar("U",33ut#,34ut#,0ut#);
pickup va_pen;
z1=(7ut,34ut);
z2=(4ut,8ut);
z3=(9ut,0ut);
z4=(24.5ut,15ut);
z5=(29ut,34ut);
z6=(22ut,5ut);
z7=(25ut,0ut);
z8=(0ut,30ut);
z9=(33ut,6ut);
z10=(10ut,31ut);
draw z8{dir 70}..z1{right}..z10{z2-z10}--z2{z2-z10}..z3{dir 15}..z4;
draw z5--z6{z6-z5}..z7{dir 10}..{dir 55}z9;
labels(range 1 thru 10);
endchar;
beginchar(oct"334",33ut#,44ut#,0ut#);   
pickup va_pen;
z1=(7ut,34ut);
z2=(4ut,8ut);
z3=(9ut,0ut);
z4=(24.5ut,15ut);
z5=(29ut,34ut);
z6=(22ut,5ut);
z7=(25ut,0ut);
z8=(0ut,30ut);
z9=(33ut,6ut);
z10=(10ut,31ut);
z11=(19.5ut,44ut);
z12=(18ut,38ut);
z13=(24.5ut,44ut);
z14=(23ut,38ut);
draw z8{dir 70}..z1{right}..z10{z2-z10}--z2{z2-z10}..z3{dir 15}..z4;
draw z5--z6{z6-z5}..z7{dir 10}..{dir 55}z9;
draw z11..z12;
draw z13..z14;
labels(range 1 thru 14);
endchar;
beginchar("V",32ut#,34ut#,0ut#);
pickup va_pen;
z1=(0ut,30ut);
z2=(7ut,34ut);
z3=(10ut,31ut);
z4=(3ut,8ut);
z5=(9ut,0ut);
z6=(20ut,10ut);
z7=(24ut,34ut);
z8=(27ut,29ut);
z9=(32ut,34ut);
draw z1{dir 70}..z2{right}..z3{z4-z3}..z4{z4-z3}..z5{dir 10}..z6{z8-z6}..z7
 & z7{down}..z8..{dir 50}z9;
labels(range 1 thru 9);
endchar;
beginchar("W",55ut#,34ut#,0ut#);
pickup va_pen;
z1=(1ut,30ut);
z2=(8ut,34ut);
z3=(11ut,31ut);
z4=(5ut,8ut);
z5=(10ut,0ut);
z6=(23ut,15ut);
z7=(28ut,34ut);
z8=(23ut,15ut);
z9=(31ut,0ut);
z10=(41ut,12ut);
z11=(41ut,34ut);
z12=(45ut,29ut);
z13=(50ut,34ut);
draw z1{dir 70}..z2{right}..z3{z4-z3}..z4{z4-z3}..z5{right}..z6{z7-z6}..z7;
draw z8{z8-z7}..z9{right}..z10{z12-z10}..z11
 & z11{down}..z12..{dir 50}z13;
labels(range 1 thru 13);
endchar;
beginchar("X",26ut#,34ut#,0ut#);
pickup va_pen;
z1=(6ut,30ut);
z2=(13ut,34ut);
z3=(18ut,22ut);
z4=(15ut,12ut);
z5=(7ut,0ut);
z6=(3ut,6ut);
z7=(16ut,16ut);
z8=(30ut,28ut);
z9=(26ut,34ut);
z10=(20ut,0ut);
z11=(26ut,6ut);
draw z1{dir 60}..z2{right}..z3{z4-z3}..z4{z4-z3}..z5{left}..z6{up}..{dir30}z7;
draw z7{dir 30}..z8{up}..z9{left}..z3{z4-z3}..z4{z4-z3}
     ..z10{dir 10}..{dir 55}z11;
labels(range 1 thru 11);
endchar;
beginchar("Y",29ut#,34ut#,15ut#);
pickup va_pen;
z1=(9ut,31ut);
z2=(1ut,8ut);
z3=(5ut,0ut);
z4=(26ut,34ut);
z5=(16ut,-5ut);
z6=(9ut,-15ut);
z7=(6ut,-12ut);
z8=(9ut,-6ut);
z9=(29ut,6ut);
z10=(0ut,30ut);
z11=(7ut,34ut);
draw z10{up}..z11{right}..z1--z2{z2-z1}..z3{right}...z4{z4-z5}
 & z4--z5{z5-z4}..z6{left}..z7..z8{z9-z8}..{dir 55}z9;
labels(range 1 thru 11);
endchar;
beginchar("Z",29ut#,34ut#,0ut#);
pickup va_pen;
z1=(3ut,29ut);
z2=(11ut,34ut);
z3=(29ut,6ut);
z4=(21ut,0ut);
z5=(23ut,31ut);
z6=(9ut,6ut);
z7=(2ut,0ut);
z8=(0ut,3ut);
z9=(29ut,34ut);
z10=(12ut,18ut);
z11=(26ut,18ut);
draw z1{dir 60}..z2{right}..z5{z9-z5}..z9
  &z9{z6-z9}--z6{z6-z9}..z7..z8..z6{z4-z6}..z4{dir 10}..{dir 55}z3;
draw z10--z11;
labels(range 1 thru 11);
endchar;
beginchar(oct"377",16ut#,34ut#,15ut#);
pickup va_pen;
z1=(-1ut*.87,19ut*.87);
z2=(3ut*.87,29ut*.87);
z3=(-7ut*.87,-18ut*.87);
z4=(13ut*.87,39ut*.87);
z5=(19ut*.87,32ut*.87);
z6=(7ut*.87,21ut*.87);
z7=(18ut*.87,8ut*.87);
z8=(10ut*.87,0ut*.87);
z9=(4ut*.87,2ut*.87);
draw z3--z2{z2-z3}..z4{right}..z5{down}..z6{left}
 & z6{right}..z7{down}..z8{left}..z9;
labels(range 1 thru 9);
endchar;
font_quad 33pt#;
font_normal_space 7.6pt#;    
font_normal_stretch 5pt#;    
font_normal_shrink 3pt#;     
font_x_height 5.7pt#;
vor_m#:=-3ut#;
vor_i#:=-1ut#;
B_e#:=-8ut#;
T_e#:=-2ut#;
boundarychar:=oct"040";
ligtable oct"025": "-"   =:  oct"026";   
ligtable  "a":oct"344":"c":"d":"e":oct"037":"f":"g":"h":"i":"j":
          "k":"l":"m":"n":"p":"t":"u":oct"374":"x":"y":"z":
              "a" |=:| 133,
          oct"344"|=:| 133,
              "b" |=:| 134,
              "c" |=:| 133,
              "d" |=:| 133,
              "f" |=:| 134,
              "g" |=:| 133,
              "h" |=:| 134,
              "i" |=:| 146,
              "j" |=:| 146,
              "k" |=:| 134,
              "l" |=:| 134,
              "m" |=:| 132,
              "n" |=:| 132,
              "o" |=:| 133,
         oct"366" |=:| 133,
              "p" |=:| 146,
              "q" |=:| 133,
              "r" |=:| 132,
              "s" |=:| 146,
         oct"377" |=:| 134,   
              "t" |=:| 134,
              "u" |=:| 146,
         oct"374" |=:| 146,
              "v" |=:| 132,
              "w" |=:| 132,
              "x" |=:| 132,
              "y" |=:| 146,
              "z" |=:| 132;
ligtable "b":"o":oct"366":"r":"v":"w":             
              "e" |=: oct"037",                    
              " " |=:| oct"043",
              "-" |=:| oct"043",
              ")" |=:| oct"043",
              "." |=:| oct"043",
              ":" |=:| oct"043",
         oct"001" |=:| oct"043",    
         oct"047" |=:| oct"043",    
         oct"140" |=:| oct"043",    
         oct"020" |=:| oct"043",    
         oct"021" |=:| oct"043",    
         oct"042" |=:| oct"043",    
              "," |=:| oct"043",
              ";" |=:| oct"043",
              "!" |=:| oct"043",
              "?" |=:| oct"043",
              "a" |=:| oct"044",
          oct"344"|=:| oct"044",
              "b" |=:| oct"045",
              "c" |=:| oct"044",
              "d" |=:| oct"044",
              "f" |=:| oct"045",
              "g" |=:| oct"044",
              "h" |=:| oct"045",
              "i" |=:| oct"043",
              "j" |=:| oct"043",
              "k" |=:| oct"045",
              "l" |=:| oct"045",
              "m" |=:| oct"046",
              "n" |=:| oct"046",
              "o" |=:| oct"044",
         oct"366" |=:| oct"044",
              "p" |=:| oct"043",
              "q" |=:| oct"044",
              "r" |=:| oct"046",
              "s" |=:| oct"043",
         oct"377" |=:| oct"043",
              "t" |=:| oct"045",
              "u" |=:| oct"043",
         oct"374" |=:| oct"043",
              "v" |=:| oct"046",
              "w" |=:| oct"046",
              "x" |=:| oct"046",
              "y" |=:| oct"043",
              "z" |=:| oct"046";
ligtable ||:"(":oct"020":oct"022":"-":     
              "-"  =: oct"025",
              "a" |=:| 130,
         oct"344" |=:| 130,
              "b" |=:| 127,
              "c" |=:| 130,
              "d" |=:| 130,
              "e" |=:| 128,
              "f" |=:| 127,
              "g" |=:| 130,
              "h" |=:| 127,
              "i" |=:| 129,
              "j" |=:| 129,
              "k" |=:| 127,
              "l" |=:| 127,
              "o" |=:| 130,
         oct"366" |=:| 130,
              "p" |=:| 129,
              "q" |=:| 130,
              "s" |=:| 129,
              "t" |=:| 127,
              "u" |=:| 129,
         oct"374" |=:| 129,
              "y" |=:| 129,
         oct"377" |=:| 127;
ligtable "s":oct"377":
              "a" |=:| 149,
          oct"344"|=:| 149,
              "b" |=:| 150,
              "c" |=:| 149,
              "d" |=:| 149,
              "e" |=:| 148,
              "f" |=:| 150,
              "g" |=:| 149,
              "h" |=:| 150,
              "i" |=:| 151,
              "j" |=:| 151,
              "k" |=:| 150,
              "l" |=:| 150,
              "m" |=:| 152,
              "n" |=:| 152,
              "o" |=:| 149,
         oct"366" |=:| 149,
              "p" |=:| 151,
              "q" |=:| 149,
              "r" |=:| 152,
              "s" |=:| 151,
              "t" |=:| 150,
              "u" |=:| 151,
         oct"374" |=:| 151,
              "v" |=:| 152,
              "w" |=:| 152,
              "x" |=:| 152,
              "y" |=:| 151,
              "z" |=:| 152;
ligtable  "A":oct"304":"C":"E":"G":"H":"J":"K":"L":"M":
          "R":"U":"X":"Y":"Z":oct"334":
              "a" |=:| 131,
         oct"344" |=:| 131,
              "b" |=:| 143,
              "c" |=:| 131,
              "d" |=:| 131,
              "f" |=:| 143,
              "g" |=:| 131,
              "h" |=:| 143,
              "i" |=:| 144,
              "j" |=:| 144,
              "k" |=:| 143,
              "l" |=:| 143,
              "m" |=:| 145,
              "n" |=:| 145,
              "o" |=:| 131,
         oct"366" |=:| 131,
              "p" |=:| 144,
              "q" |=:| 131,
              "s" |=:| 144,
              "r" |=:| 145,
              "t" |=:| 143,
              "u" |=:| 144,
         oct"374" |=:| 144,
              "v" |=:| 145,
              "w" |=:| 145,
              "x" |=:| 145,
              "y" |=:| 144,
              "z" |=:| 145;
ligtable  "B":"D":"O":oct"326":"V":"W":"N":
              "a" |=:| 139,
         oct"344" |=:| 139,
              "b" |=:| 140,
              "c" |=:| 139,
              "d" |=:| 139,
              "e" |=:| 153,
              "f" |=:| 140,
              "g" |=:| 139,
              "h" |=:| 140,
              "i" |=:| 141,
              "j" |=:| 141,
              "k" |=:| 140,
              "l" |=:| 140,
              "m" |=:| 142,
              "n" |=:| 142,
              "o" |=:| 139,
         oct"366" |=:| 139,
              "p" |=:| 142,
              "q" |=:| 139,
              "r" |=:| 142,
              "s" |=:| 141,
              "t" |=:| 140,
              "u" |=:| 141,
         oct"374" |=:| 141,
              "v" |=:| 142,
              "w" |=:| 142,
              "x" |=:| 142,
              "y" |=:| 141,
              "z" |=:| 142,
             153 kern B_e#,
             141 kern vor_i#,
             142 kern vor_m#;
ligtable  "F":"I":"P":"S":"T":
              "a" |=:| 135,
         oct"344" |=:| 135,
              "b" |=:| 138,
              "c" |=:| 135,
              "d" |=:| 135,
              "e" |=:| 147,
              "f" |=:| 138,
              "g" |=:| 135,
              "h" |=:| 138,
              "i" |=:| 136,
              "j" |=:| 136,
              "k" |=:| 138,
              "l" |=:| 138,
              "m" |=:| 137,
              "n" |=:| 137,
              "o" |=:| 135,
         oct"366" |=:| 135,
              "p" |=:| 136,
              "q" |=:| 135,
              "r" |=:| 137,
              "s" |=:| 136,
              "t" |=:| 138,
              "u" |=:| 136,
         oct"374" |=:| 136,
              "v" |=:| 137,
              "w" |=:| 137,
              "x" |=:| 137,
              "y" |=:| 136,
              "z" |=:| 137,
              147 kern T_e#,
              137 kern vor_m#;
bye.
}%
\catcode`#=6
\immediate\closeout\lafont%
\usepackage{version}

\includeversion{IPL}
\def\authorrunning#1{}
\def\institute#1{}
\def\inst#1{}
\newcounter{df}
\def\definitionname{Definition}
\newtheorem{definition}[df]{Definition}
\newcounter{thm}
\def\theoremname{Theorem}
\newtheorem{theorem}[thm]{Theorem}

\newcounter{lem}
\def\lemmaname{Lemma}
\newtheorem{lemma}[lem]{Lemma}
\newcounter{cor}
\def\corollaryname{Corollary}
\newtheorem{corollary}[cor]{Corollary}
\newenvironment{proof}{\begin{trivlist}\item[\hspace{\labelsep}\bf Proof:]}{\end{trivlist}}
\def\bbbn{{\rm I\!N}} 
\let\oldmaketitle\maketitle
\def\maketitle{}
\let\oldendabstract\endabstract
\def\endabstract{\oldendabstract\oldmaketitle}

\usepackage{graphicx}
\usepackage{pstricks}
\usepackage{amsmath}
\usepackage{amssymb}
\usepackage{stmaryrd}
\usepackage{mathrsfs}
\usepackage{version}

\def\text#1{\textrm{#1}}

\def\precond#1{{\vphantom{#1}}^\bullet #1}
\def\postcond#1{{#1}^\bullet}

\def\production#1{\stackrel{#1}{\longrightarrow}}

\newfont{\fsc}{eusm10 scaled 1100}      

\def\powermultiset#1{\bbbn^{#1}}
\def\implies{\Rightarrow}

\def\mathrlap{\mathpalette\mathrlapinternal}
\def\mathrlapinternal#1#2{%
  \rlap{$\mathsurround=0pt#1{#2}$}}
\def\mathllap{\mathpalette\mathllapinternal}
\def\mathllapinternal#1#2{%
  \llap{$\mathsurround=0pt#1{#2}$}}
\def\trail#1{\text{~#1}}

\def\defitem#1{\emph{#1}}

\let\origexists\exists
\let\orignexists\nexists
\let\origforall\forall
\def\quantorSpace{}
\def\exists#1.{\quantorSpace\origexists\def\quantorSpace{\,}#1.\onespace}
\def\nexists#1.{\quantorSpace\orignexists\def\quantorSpace{\,}#1.\onespace}
\def\forall#1.{\quantorSpace\origforall\def\quantorSpace{\,}#1.\onespace}
\def\onespace#1{\let\argument=#1\ifx\onespace#1\else~\fi\argument}

\let\origmin\min
\def\min{\mathord{\origmin}}
\let\origmax\max
\def\max{\mathord{\origmax}}

\def\quireunderscore{_}
\def\quire#1{%
  \def\tmp{#1}%
  \ifx\tmp\quireunderscore%
    \def\tmp{\quireindexed_}
  \else%
    \def\tmp{\mathscr{Q}#1}
  \fi\tmp}
\def\quireindexed_#1{\mathscr{Q}_{\text{#1}}}

\newcommand{\refdf}[1]{\definitionname~\ref{df-#1}}

\newcommand{\refthm}[1]{\theoremname~\ref{thm-#1}}
\newcommand{\refcor}[1]{\corollaryname~\ref{cor-#1}}
\newcommand{\reflem}[1]{\lemmaname~\ref{lem-#1}}
\newcommand{\reffig}[1]{\figurename~\ref{fig-#1}}

\newcommand{\refsec}[1]{Section~\ref{sec-#1}}

\makeatletter
\def\goesto{\@transition\rightarrowfill}
\def\Goesto{\@transition\Rightarrowfill}
\def\ngoesto{\@transition\nrightarrowfill}
\def\nGoesto{\@transition\nRightarrowfill}
\def\@transition#1{\@ifnextchar[{\@@transition{#1}}{\@@transition{#1}[]}}
\newbox\@transbox
\newbox\@arrowbox
\def\rightarrowfill{$\m@th\mathord-\mkern-6mu%
  \cleaders\hbox{$\mkern-2mu\mathord-\mkern-2mu$}\hfill
  \mkern-6mu\mathord\rightarrow$}
\def\Rightarrowfill{$\m@th\mathord=\mkern-6mu%
  \cleaders\hbox{$\mkern-2mu\mathord=\mkern-2mu$}\hfill
  \mkern-6mu\mathord\Rightarrow$}
\def\@@transition#1[#2]%
   {\setbox\@transbox\hbox
       {\vrule height 1.5ex depth 1ex width 0ex\hskip0.25em$\scriptstyle#2$\hskip0.25em}
   \ifdim\wd\@transbox<1.5em
      \setbox\@transbox\hbox to 1.5em{\hfil\box\@transbox\hfil}\fi
   \setbox\@arrowbox\hbox to \wd\@transbox{#1}
   \ht\@arrowbox\z@\dp\@arrowbox\z@
   \setbox\@transbox\hbox{$\mathop{\box\@arrowbox}\limits^{\box\@transbox}$}
   \ht\@transbox\z@\dp\@transbox\z@
   \mathrel{\box\@transbox}}
\makeatother

\def\alignedcaption[#1&#2]{\mbox{\scriptsize $\mathllap{#1{}}\mathrlap{#2}$}}

\def\ie{i.e.\ }

\def\varnothing{\emptyset}

\def\restrictedto{\mathop\upharpoonright}

\newcommand{\inp}{\mathbin\in}                        
\def\idx#1#2#3#4#5{
  \def\argone{#1}
  \def\argtwo{#2}
  \def\argthree{#3}
  \def\argfour{#4}
  \def\argfive{#5}
  \def\testprime{'}
  \def\testdprime{''}
  \def\testtprime{'''}
  {\vphantom{\argthree}}_%
    {\vphantom{\argfour}\argone}^%
    {\vphantom{\argfive}{\argtwo}}%
  \argthree_%
    {\vphantom{\argone}\argfour}%
    \ifx\argfive\testprime\argfive\else%
    \ifx\argfive\testdprime\argfive\else%
    \ifx\argfive\testtprime\argfive\else%
    ^{\vphantom{\argtwo}\argfive}\fi\fi\fi%
}

\newcommand{\monus}{\mathrel{\raisebox{-0pt}[0pt][0pt]{$
                      \stackrel{\raisebox{-5pt}[0pt][0pt]{\huge$\cdot$}}
                               {\raisebox{0pt}[0pt][0pt]{$-$}}$}}}

\def\FS{{\it FS}}

\def\swap{\mbox{swap}}
\def\swapeq{\approx_{\mbox{s}}}
\def\adjacent{\mathrel{\,\,\rput(0,0.1){\psscalebox{-0.8}{\leftrightarrow}}\,\,}}
\def\connectedto{\adjacent^*}
\def\Lin{\text{Lin}}
\def\FS{\text{FS}}
\def\boldbracketa{%
  \psline[linewidth=0.1pt]{-}(0.15,-0.05)(0.05,-0.05)%
  \psline[linewidth=1.5pt]{-}(0.05,-0.05)(0.05,0.25)%
  \psline[linewidth=0.1pt]{-}(0.05,0.25)(0.15,0.25)%
  \phantom{[\,}%
}
\def\boldbracketb{%
  \psline[linewidth=0.1pt]{-}(0.0,-0.05)(0.1,-0.05)%
  \psline[linewidth=1.5pt]{-}(0.1,-0.05)(0.1,0.25)%
  \psline[linewidth=0.1pt]{-}(0.1,0.25)(0.0,0.25)%
  \phantom{]\,}%
}
\def\BD#1{\mathop{\boldbracketa#1\boldbracketb}}

\def\BDify#1{{\it BD}(#1)}
\newenvironment{itemise}{\begin{list}{$\bullet$}{\leftmargin 12pt \labelwidth\leftmargin\advance\labelwidth-\labelsep \topsep 4pt \itemsep 2pt \parsep 2pt}}{\end{list}}
\newenvironment{itemisei}{\begin{list}{$-$}{\leftmargin 12pt \labelwidth\leftmargin\advance\labelwidth-\labelsep \topsep 4pt \itemsep 2pt \parsep 2pt}}{\end{list}}
\def\justempty{}
\newenvironment{define}[1]{\begin{definition}\rm\def\arga{#1}\ifx\justempty\arga~\else#1\fi\\\vspace{-6ex}\\\mbox{~}\begin{itemise}}{\end{itemise}\end{definition}}

\DeclareFontFamily{T1}{la}{}
\DeclareFontShape{T1}{la}{m}{n}{<->s*[0.8571428571]la14}{}
\DeclareRobustCommand\la{\fontfamily{la}\fontencoding{T1}\selectfont}
\def\processfont#1{\text{\la #1}}
\def\NN{\processfont{N}}
\def\SS{\processfont{S}\,}
\def\TT{\processfont{T}}
\def\FF{\processfont{F}}
\def\MM{\processfont{M}}
\def\PP{P}
\def\QQ{Q}

\usepackage{pstricks}
\usepackage{pst-node}
\usepackage{graphicx}

\newcounter{netimage}
\def\p#1:#2;{\cnode #1{0.3}{n\thenetimage-#2}}
\def\P#1:#2;{\p #1:#2;\pscircle*#1{0.1}}
\def\q#1:#2:#3;{\p #1:#2;\rput#1{\rput[l](0.45,0){\large\it #3}}}
\def\Q#1:#2:#3;{\P #1:#2;\rput#1{\rput[l](0.45,0){\large\it #3}}}
\def\qq#1:#2:#3;{\p #1:#2;\rput#1{\rput[t](0,-0.5){\large\it #3}}}
\def\ql#1:#2:#3;{\p #1:#2;\rput#1{\rput[r](-0.45,0){\large\it #3}}}
\def\qt#1:#2:#3;{\p #1:#2;\rput#1{\rput[b](0,0.45){\large\it #3}}}
\def\Qt#1:#2:#3;{\P #1:#2;\rput#1{\rput[b](0,0.45){\large\it #3}}}
\def\Ql#1:#2:#3;{\P #1:#2;\rput#1{\rput[r](-0.45,0){\large\it #3}}}
\def\qx#1:#2:#3:#4;{\p #1:#2;\rput#1{\rput#4{\large\it #3}}}
\def\QXX#1:#2:#3:#4:#5;{\p #1:#2;\rput#1{\rput#4{\large\it #3}}\pscircle*#5{0.1}}
\def\s#1:#2:#3;{\p #1:#2;\rput#1{\rput(-0.03,0){\large\it #3}}}
\def\t#1:#2:#3;{\rput#1{\rnode{n\thenetimage-#2}{\psframebox{%
  \vbox to 0.6cm{\vfil\hbox to 0.6cm{\hfil\Large\it #3\hfil}\vfil}}}}}
\def\u#1:#2:#3:#4;{\rput#1{\rnode{n\thenetimage-#2}{\psframebox{%
  \vbox to 0.6cm{\vfil\hbox to 0.6cm{\hfil\Large\it #3\hfil}\vfil}}}}%
  \rput#1{\rput[l](0.6,0){\large\it #4}}}
\def\ut#1:#2:#3:#4;{\rput#1{\rnode{n\thenetimage-#2}{\psframebox{%
  \vbox to 0.6cm{\vfil\hbox to 0.6cm{\hfil\Large\it #3\hfil}\vfil}}}}%
  \rput#1{\rput[b](0,0.6){\large\it #4}}}
\def\ul#1:#2:#3:#4;{\rput#1{\rnode{n\thenetimage-#2}{\psframebox{%
  \vbox to 0.6cm{\vfil\hbox to 0.6cm{\hfil\Large\it #3\hfil}\vfil}}}}%
  \rput#1{\rput[r](-0.6,0){\large\it #4}}}
\def\a#1->#2;{\ncline{->}{n\thenetimage-#1}{n\thenetimage-#2}}
\def\A#1->#2;{\ncarc{->}{n\thenetimage-#1}{n\thenetimage-#2}}
\def\B#1->#2;{\ncarc[arcangle=-8]{->}{n\thenetimage-#1}{n\thenetimage-#2}}
\def\avlinearc{0.2}
\def\av#1[#2]-#3->[#4]#5;{
  \SpecialCoor
  \psline[linearc=\avlinearc]{->}([angle=#2]n\thenetimage-#1)#3([angle=#4]n\thenetimage-#5)
}

\makeatletter
\long\def\petrinet(#1)#2\end{\psscalebox{0.7}{\pspicture(#1)\stepcounter{netimage}#2\endpspicture}\end}

\makeatother

\begin{document}

\begin{IPL}
\title{Abstract Processes of Place/Transition Systems\tnoteref{dfg}}
\tnotetext[dfg]{This work was partially supported by the DFG (German Research Foundation).}

\author[nicta,unsw]{Rob van Glabbeek}
\ead{rvg@cs.stanford.edu}

\author[ips]{Ursula Goltz}
\ead{goltz@ips.cs.tu-bs.de}

\author[ips]{Jens-Wolfhard Schicke}
\ead{drahflow@gmx.de}

\address[nicta]{NICTA, Sydney, Australia}
\address[unsw]{School of Computer Science and Engineering, University of New South Wales, Sydney, Australia}
\address[ips]{Institute for Programming and Reactive Systems, TU Braunschweig, Germany}
\end{IPL}

\maketitle

\begin{abstract}
A well-known problem in Petri net theory is to formalise an
appropriate causality-based concept of process or run for
place/transition systems. The so-called individual token
interpretation, where tokens are distinguished according to their
causal history, giving rise to the processes of Goltz and Reisig, is
often considered too detailed.  The problem of defining a fully
satisfying more abstract concept of process for general
place/transition systems has so-far not been solved.  In this paper,
we recall the proposal of defining an abstract notion of process, here
called \defitem{BD-process}, in terms of equivalence classes of
Goltz-Reisig processes, using an equivalence proposed by Best and
Devillers. It yields a fully satisfying solution for at least all
one-safe nets. However, for certain nets which intuitively have
different conflicting behaviours, it yields only one maximal abstract
process.  Here we identify a class of place/transition systems, called
\defitem{structural conflict nets}, where conflict and concurrency
due to token multiplicity are clearly separated.
We show that, in the case of structural conflict nets, the equivalence
proposed by Best and Devillers yields a unique maximal abstract
process only for conflict-free nets.
Thereby BD-processes constitute a simple and fully
satisfying solution in the class of structural conflict nets.

\begin{IPL}
\begin{keyword}
  Petri nets, \mbox{P\hspace{-1pt}/T} systems, causal semantics, processes
\end{keyword}
\end{IPL}
\end{abstract}

\section{Introduction}\label{sec-intro}
\noindent
The most frequently used class of Petri nets are \defitem{place/}
\defitem{transition systems} (\mbox{P\hspace{-1pt}/T} systems) where
places may carry arbitrary many tokens, or a certain maximal number of
tokens when adding place capacities. These tokens are usually assumed
to be indistinguishable entities. Multiplicities of tokens may
represent for instance the number of available resources in a
system. The semantics of this type of Petri nets is well-defined with
respect to single firings of transitions or finite sets of transitions
firing in parallel (\defitem{steps}).  Sequences of transition firings
or of steps are the usual way to define the behaviour of a
\mbox{P\hspace{-1pt}/T} system. However, these notions of behaviour do
not fully reflect the power of Petri nets, as they do not explicitly
represent causal dependencies between transition occurrences.  If one
wishes to interpret \mbox{P\hspace{-1pt}/T} systems with a causal
semantics, several interpretations of what ``causal semantics'' should
actually mean are available. In the following we give a short overview.

Initially, Petri introduced the concept of a net together with the
definition of the firing rule.  He defined \defitem{condition/event
systems}, where---amongst other restrictions---places (then called
conditions) may carry at most one token. For this class of nets, he
proposed what is now the classical notion of a \defitem{process},
given as a mapping from an \defitem{occurrence net} (acyclic net with
unbranched places) to the original net \cite{petri77nonsequential,genrich80dictionary}.
A process models a run of the represented system, obtained by
choosing one of the alternatives in case of conflict. It records all
occurrences of the places and transitions visited during such a run,
together with the causal dependencies between them, which are given by
the flow relation of the net.
A linear-time causal semantics of a condition/event system is thus
obtained by associating with a net the set of its processes.
Depending on the desired level of abstraction, it may suffice to
extract from each process just the partial order of transition
occurrences in it. The firing sequences of transitions or steps can in
turn be extracted from these partial orders.
Nielsen, Plotkin and Winskel extended this to a branching-time semantics by
using occurrence nets with forward branched places \cite{nielsen81petri}.
These capture all runs of the represented system, together with the
branching structure of choices between them.

Goltz and Reisig generalised Petri's notion of process to general
\mbox{P\hspace{-1pt}/T} systems where multiple tokens may reside on a single place
\cite{goltz83nonsequential}. We call this notion of a process \defitem{GR-process}.
Engelfriet adapted GR-processes by additionally representing choices
between alternative behaviours \cite{engelfriet91branchingprocesses},
thereby adopting the approach of \cite{nielsen81petri} to
\mbox{P\hspace{-1pt}/T} systems, although without arc weights.
Meseguer, Sassone and Montanari extended this to cover also arc weights \cite{MMS97}.

However, Goltz argued that when abstracting from the identity of
multiple tokens residing in the same place, GR-processes do not
accurately reflect runs of nets, because if a Petri net is
conflict-free it should intuitively have only one run (for there are
no choices to resolve), yet it may have multiple GR-processes
\cite{goltz86howmany}.
This phenomenon is illustrated in \reffig{unsafe} in \refsec{semantics}.
A similar argument is made, e.g., in \cite{HKT95}.

At the heart of this issue is the question whether multiple tokens
residing in the same place should be seen as individual entities, so
that a transition consuming just one of them constitutes a conflict, or
whether such tokens are indistinguishable, so that taking one is
equivalent to taking the other.
Van Glabbeek and Plotkin call the former viewpoint the
\defitem{individual token interpretation} of P\hspace{-1pt}/T systems
and the latter the
\defitem{collective token interpretation} \cite{glabbeek95configuration}.
A formalisation of these interpretations occurs in \cite{glabbeek05individual}.
A third option, proposed by Vogler, regards
tokens only as notation for a natural number stored in each place;
these numbers are incremented or decremented when firing transitions, thereby
introducing explicit causality between any transitions removing tokens from
the same place \cite{vogler91executions}.
The GR-processes, as well as the work of
\cite{engelfriet91branchingprocesses,MMS97}, fit with the individual
token interpretation.

In this paper we continue the line of research of
\cite{goltz86howmany,MM88,DMM96,mazurkiewicz89multitree,HKT95}
to formalise a causality-based notion of run of a Petri net that
fits the collective token interpretation.
As remarked already in \cite{goltz86howmany},\linebreak
{\it what we need is some notion of an ``abstract
process''} and {\it a notion of maximality for abstract processes},
such that\linebreak[3]
{\it a \mbox{P\hspace{-1pt}/T}-system is conflict-free iff it has
exactly one maximal abstract process starting at the initial marking.}

A canonical candidate for such a notion of an abstract process is an
equivalence class of GR-processes, using an equivalence notion
($\equiv^\infty_{1}\hspace{-.5pt}$) proposed by Best and Devillers
\cite{best87both}.
This equivalence relation is generated by a transformation for
changing causalities in GR-processes, called \defitem{swapping}, that
identifies GR-processes which differ only in the choice which token
was removed from a place.  Here we call the resulting notion of a more
abstract process \defitem{BD-process}.  In the special case of
one-safe \mbox{P\hspace{-1pt}/T} systems (where places carry at most
one token), or for condition/event systems, no swapping is possible,
and a BD-process is just an isomorphism class of GR-processes.

Meseguer and Montanari formalise runs in a net $N$ as morphisms in a
category $\mathcal{T}(N)$ \cite{MM88}. In \cite{DMM96} it has been
established that these morphisms ``coincide with the commutative
processes defined by Best and Devillers'' (their terminology for
BD-processes). Likewise, Hoogers, Kleijn and Thiagarajan represent an
abstract run of a net by a \defitem{trace}, thereby generalising the
trace theory of Mazurkiewicz \cite{mazurkiewicz95tracetheory}, and
remark that ``it is straightforward but laborious to set up a 1-1
correspondence between our traces and the equivalence classes of
finite processes generated by the swap operation in [Best and
Devillers, 1987].''.  Mazurkiewicz applies a different approach with
his \defitem{multitrees} \cite{mazurkiewicz89multitree}, which record
possible multisets of fired transitions.  This approach applies to
nets without self-loops only, and we will not consider it in this
paper.

Best and Devillers have shown that their equivalence classes of
GR-processes are in a bijective correspondence with equivalence
classes of firing sequences, generated by swapping two adjacent
transitions firings that could have been done in one step.
This gives further evidence for the suitability of BD-processes as a
formalisation of abstract runs.
However, it can be argued that this solution is not fully satisfying
for general \mbox{P\hspace{-1pt}/T} systems, as we will recall in \refsec{semantics} using
an example from Ochma\'nski \cite{ochmanski89personal}. It identifies
GR-processes in such a way that certain \mbox{P\hspace{-1pt}/T} systems
with conflicts have only one maximal BD-process.

In this paper, we analyse the notion of conflict in \mbox{P\hspace{-1pt}/T} systems and
its interplay with concurrency and causality. We recall the definition
of the notion of conflict for \mbox{P\hspace{-1pt}/T} systems from
\cite{goltz86howmany}. We then define a subclass of \mbox{P\hspace{-1pt}/T} systems,
called \defitem{structural conflict nets}, where the interplay between
conflicts and concurrency due to token multiplicities is clearly
separated. On this class, the notions of syntactic and semantic
conflict are in complete agreement.
 We show that, for this subclass of \mbox{P\hspace{-1pt}/T} systems, the swap transformation by Best and Devillers yields a unique maximal BD-process only for those nets which are conflict-free. The proof of this result is quite involved; it is achieved by using the alternative characterisation of BD-processes by firing sequences from \cite{best87both}.

We proceed by defining basic notions for \mbox{P\hspace{-1pt}/T}
systems in Section \ref{sec-basic}.  In Section \ref{sec-semantics},
we define GR-processes and introduce the swapping equivalence. We give
examples and discuss the deficiencies of both GR-processes and
BD-processes for a collective token interpretation of general
\mbox{P\hspace{-1pt}/T} systems.  Section \ref{sec-conflict}
recapitulates the concept of conflict in \mbox{P\hspace{-1pt}/T}
systems and defines structural conflict nets. In Sections
\ref{sec-finiteruns} and \ref{sec-conflictruns}, respectively, we
introduce the alternative characterisation of BD-processes from
\cite{best87both} in terms of equivalence classes of firing sequences
and prove in this setting that structural conflict nets with a unique
maximal run are indeed conflict-free.
Finally we transfer the result to BD-processes in Section \ref{sec-results}.

\section{Place/Transition Systems}\label{sec-basic}

\noindent
We will employ the following notations for multisets.

\begin{define}{
  Let $X$ be a set.
}\label{df-multiset}
\item A {\em multiset} over $X$ is a function $A\!:X \rightarrow \bbbn$,
i.e.\ $A\in \powermultiset{X}\!\!$.
\item $x \in X$ is an \defitem{element of} $A$, notation $x \in A$, iff $A(x) > 0$.
\item For multisets $A$ and $B$ over $X$ we write $A \subseteq B$ iff
 \mbox{$A(x) \leq B(x)$} for all $x \inp X$;
\\ $A\cup B$ denotes the multiset over $X$ with $(A\cup B)(x):=\text{max}(A(x), B(x))$,
\\ $A + B$ denotes the multiset over $X$ with $(A + B)(x):=A(x)+B(x)$,
\\ $A - B$ is given by
$(A - B)(x):=A(x)\monus B(x)=\\\mbox{max}(A(x)-B(x),0)$, and\\
for $k\inp\bbbn$ the multiset $k\cdot A$ is given by
$(k \cdot A)(x):=k\cdot A(x)$.
\item The function $\emptyset\!:X\rightarrow\bbbn$, given by
  $\emptyset(x):=0$ for all $x \inp X$, is the \emph{empty} multiset over $X$.
\item If $A$ is a multiset over $X$ and $Y\subseteq X$ then
  $A\restrictedto Y$ denotes the multiset over $Y$ defined by
  $(A\restrictedto Y)(x) := A(x)$ for all $x \inp Y$.
\item The cardinality $|A|$ of a multiset $A$ over $X$ is given by
  $|A| := \sum_{x\in X}A(x)$.
\item A multiset $A$ over $X$ is \emph{finite}
  iff $|A|<\infty$, i.e.,
  iff the set $\{x \mid x \inp A\}$ is finite.
\end{define}
Two multisets $A\!:X \rightarrow \bbbn$ and
$B\!:Y\rightarrow \bbbn$
are \emph{extensionally equivalent} iff
$A\restrictedto (X\cap Y) = B\restrictedto (X\cap Y)$,
$A\restrictedto (X\setminus Y) = \emptyset$, and
$B \restrictedto (Y\setminus X) = \emptyset$.
In this paper we often do not distinguish extensionally equivalent
multisets. This enables us, for instance, to use $A \cup B$ even
when $A$ and $B$ have different underlying domains.
With $\{x,x,y\}$ we will denote the multiset over $\{x,y\}$ with
$A(x)\mathbin=2$ and $A(y)\mathbin=1$, rather than the set $\{x,y\}$ itself.
A multiset $A$ with $A(x) \leq 1$ for all $x$ is
identified with the set $\{x \mid A(x)=1\}$.

Below we define place/transition systems as net structures with an initial marking.
In the literature we find slight variations in the definition of \mbox{P\hspace{-1pt}/T}
systems concerning the requirements for pre- and postsets of places
and transitions. In our case, we do allow isolated places. For
transitions we allow empty postsets, but require at least one
preplace, thus avoiding problems with infinite self-concurrency.
Moreover, following \cite{best87both}, we restrict attention
to nets of \defitem{finite synchronisation}, meaning that each
transition has only finitely many pre- and postplaces.
Arc weights are included by defining the flow relation as a function to the natural numbers.
For succinctness, we will refer to our version of a \mbox{P\hspace{-1pt}/T} system as a \defitem{net}.

\begin{define}{}\label{df-nst}
\item[]
  A \defitem{net} is a tuple
  $N = (S, T, F, M_0)$ where
  \begin{itemise}
    \item $S$ and $T$ are disjoint sets (of \defitem{places} and \defitem{transitions}),
    \item $F: (S \mathord\times T \mathrel\cup T \mathord\times S) \rightarrow \bbbn$
      (the \defitem{flow relation} including \defitem{arc weights}), and
    \item $M_0 : S \rightarrow \bbbn$ (the \defitem{initial marking})
  \end{itemise}
  such that for all $t \inp T$ the set $\{s\mid F(s, t) > 0\}$ is
  finite and non-empty, and the set $\{s\mid F(t, s) > 0\}$ is finite.
\end{define}

\noindent
Graphically, nets are depicted by drawing the places as circles and
the transitions as boxes. For $x,y \inp S\cup T$ there are $F(x,y)$
arrows (\defitem{arcs}) from $x$ to $y$.  When a net represents a
concurrent system, a global state of this system is given as a
\defitem{marking}, a multiset of places, depicted by placing $M(s)$
dots (\defitem{tokens}) in each place $s$.  The initial state is
$M_0$.  The system behaviour is defined by the possible moves between
markings $M$ and $M'$, which take place when a finite multiset $G$ of
transitions \defitem{fires}.  When firing a transition, tokens on
preplaces are consumed and tokens on postplaces are created, one for
every incoming or outgoing arc of $t$, respectively.  Obviously, a
transition can only fire if all necessary tokens are available in $M$
in the first place. \refdf{firing} formalises this notion of
behaviour.

\begin{define}{
  Let $N\!=\!(S, T, F, M_0)$ be a net and $x\inp S\cup T$.
}\label{df-preset}
\item[]
The multisets $\precond{x},~\postcond{x}: S\cup T \rightarrow
\bbbn$ are given by $\precond{x}(y)=F(y,x)$ and
$\postcond{x}(y)=F(x,y)$ for all $y \inp S \cup T$.
If $x\in T$, the elements of $\precond{x}$ and $\postcond{x}$ are
called \emph{pre-} and \emph{postplaces} of $x$, respectively.
These functions extend to multisets
$X:S \cup T \rightarrow\bbbn$ as usual, by
$\precond{X} := \Sigma_{x \in S \cup T}X(x)\cdot\precond{x}$ and
$\postcond{X} := \Sigma_{x \in S \cup T}X(x)\cdot\postcond{x}$.
\end{define}

\begin{define}{
  Let $N \mathbin= (S, T, F, M_0)$ be a net,
  $G \in \bbbn^T\!$, $G$ non-empty and finite, and $M, M' \in \bbbn^S\!$.
}\label{df-firing}
\item[]
$G$ is a \defitem{step} from $M$ to $M'$,
written $M\production{G}_N M'$, iff
\begin{itemise}
  \item $^\bullet G \subseteq M$ ($G$ is \defitem{enabled}) and
  \item $M' = (M - \mbox{$^\bullet G$}) + G^\bullet$. 
\end{itemise}
We may leave out the subscript $N$ if clear from context.
Extending the notion to words $\sigma = t_1t_2\ldots t_n \in T^*$
we write $M\production{\sigma} M'$ for\vspace{-5pt}
$$
\exists M_1, M_2, \ldots, M_{n-1}.
M\!\production{\{t_1\}}\! M_1\!\production{\{t_2\}}\! M_2 \cdots M_{n-1}\!\production{\{t_n\}}\! M'\!\!.
$$
When omitting $\sigma$ or $M'$ we always mean it to be existentially quantified.
When $M_0 \production{\sigma}_N$, the word $\sigma$
is called a \defitem{firing sequence} of $N$.
The set of all firing sequences of $N$ is denoted by $\FS(N)$.
\end{define}

\noindent
Note that steps are (finite) multisets, thus allowing self-concurrency.
Also note that $M\goesto[\{t,u\}]$ implies $M\goesto[tu]$ and $M\goesto[ut]$.
We use the notation $t\in \sigma$ to indicate that the transition $t$
occurs in the sequence $\sigma$, and $\sigma\leq\rho$ to indicate that
$\sigma$ is a prefix of the sequence $\rho$, i.e.\ $\exists \mu. \rho=\sigma\mu$.

\section{Processes of Place/Transition Systems}\label{sec-semantics}

\noindent
We now define processes of nets.
A (GR-)process is essentially a conflict-free, acyclic net together
with a mapping function to the original net. It can be obtained by
unwinding the original net, choosing one of the alternatives in case
of conflict.
The acyclic nature of the process gives rise to a notion of causality
for transition firings in the original net via the mapping function.
Conflicts present in the original net are represented by one net yielding
multiple processes, each representing one possible way to decide the conflicts.

\begin{define}{}\label{df-process}
 \item[]
  A pair $\PP = (\NN, \pi)$ is a
  \defitem{(GR-)process} of a net\\ $N = (S, T, F, M_0)$
  iff
  \begin{itemise}\itemsep 3pt
   \item $\NN = (\SS, \TT, \FF, \MM_0)$ is a net, satisfying
   \begin{itemisei}
    \item $\forall s \in \SS. |\precond{s}| \leq\! 1\! \geq |\postcond{s}|
    \wedge\, \MM_0(s) = \left\{\begin{array}{@{}l@{\quad}l@{}}1&\mbox{if $\precond{s}=\emptyset$}\\
                                   0&\mbox{otherwise,}\end{array}\right.$
    \item $\FF$ is acyclic, \ie
      $\forall x \inp \SS \cup \TT. (x, x) \mathbin{\not\in} \FF^+$,
      where $\FF^+$ is the transitive closure of $\{(t,u)\mid F(t,u)>0\}$,
    \item and $\{t \mid (t,u)\in \FF^+\}$ is finite for all $u\in \TT$.
   \end{itemisei}
    \item $\pi:\SS \cup \TT \rightarrow S \cup T$ is a function with 
    $\pi(\SS) \subseteq S$ and $\pi(\TT) \subseteq T$, satisfying
   \begin{itemisei}
    \item $\pi(\MM_0) = M_0$, i.e.\ $M_0(s) = |\pi^{-1}(s) \cap \MM_0|$ for all $s\in S$, and
    \item $\forall t \in \TT, s \in S.
      F(s, \pi(t)) = |\pi^{-1}(s) \cap \precond{t}| \wedge\\
      F(\pi(t), s) = |\pi^{-1}(s) \cap \postcond{t}|$.
  \end{itemisei}
  \end{itemise}
  $P$ is called \defitem{finite} if $\TT$ (and  hence $\SS$) are finite.
\end{define}

\noindent
The conditions for $\NN$ ensure that a process is indeed a mapping from an occurrence net as defined in \cite{petri77nonsequential,genrich80dictionary} to the net $N$; hence we define processes here in the classical way as in \cite{goltz83nonsequential,best87both} 
(even though not introducing occurrence nets explicitly).

A process is not required to represent a completed run of the original net.
It might just as well stop early. In those cases, some set of
transitions can be added to the process such that another (larger) process is
obtained. This corresponds to the system taking some more steps and gives rise
to a natural order between processes.

\begin{define}{
  Let $\PP = ((\SS, \TT, \FF, \MM_0), \pi)$ and\\ $\PP' = ((\SS', \TT', \FF', \MM_0'), \pi')$ be
  two processes of the same net.
}\label{df-extension}
\item
  $\PP'$ is a \defitem{prefix} of $\PP$, notation $\PP'\leq \PP$, and 
  $\PP$ an \defitem{extension} of $\PP'$, iff 
    $\SS'\subseteq \SS$,
    $\TT'\subseteq \TT$,
    $\MM_0' = \MM_0$,
    $\FF'=\FF\restrictedto(\SS' \mathord\times \TT' \mathrel\cup \TT' \mathord\times \SS')$
    and $\pi'=\pi\restrictedto(\SS'\times \TT')$.
\item
  A process of a net is said to be \defitem{maximal} if 
  it has no proper extension.
\end{define}

\noindent
The requirements above imply that if $\PP'\leq \PP$, $(x,y)\in
\FF^+$ and $y\in \SS' \cup \TT'$ then $x\in \SS' \cup \TT'$.
Conversely, any subset $\TT'\subseteq \TT$ satisfying
$(t,u)\in \FF^+ \wedge u\in \TT' \Rightarrow t\in \TT'$ uniquely determines a
prefix of $\PP$.

Two processes $(\NN, \pi)$ and $(\NN', \pi')$
are \defitem{isomorphic} iff there exists an isomorphism $\phi$ from
$\NN$ to $\NN'$ which respects the process mapping, i.e.\
$\pi = \pi' \circ \phi$.
Here an isomorphism $\phi$ between two nets $\NN=(\SS, \TT, \FF,
\MM_0)$ and $\NN'=(\SS', \TT', \FF', \MM'_0)$ is a
bijection between their places and transitions such that
$\MM'_0(\phi(s))=\MM_0(s)$ for all $s\in\SS$ and
$F'(\phi(x),\phi(y))=F(x,y)$ for all $x,y\in \SS\cup\TT$.

The notion of a GR-process presented above
may fail to capture the intuitive concept of an
abstract run of the represented system
if the original net (e.g.\ the one in \reffig{unsafe})
reaches a marking with multiple tokens in one place.
According to \refdf{process}, such a net $N$ has processes
in which multiple places are mapped to the same place in $N$,
thereby representing multiple tokens there. (In \reffig{unsafe}
each of the two represented processes has two places that map to place
4 in $N$.)
If such a process features a transition whose counterpart in $N$
consumes just one of the tokens present in this place,
one needs to choose which of the equivalent places in the
process to connect to this transition.
Taking different choices gives rise to different processes
(such as the two in \reffig{unsafe}).
There is the philosophical question of whether multiple tokens
in the same place are distinct entities, that may induce distinct
causal relationships---the individual token interpretation,
or together constitute some state of the place---the collective token
interpretation \cite{glabbeek95configuration}.
In the collective token interpretation of \mbox{P\hspace{-1pt}/T}
systems, which we take in this paper, the choice which token to remove
should not lead to different runs.

\begin{figure}
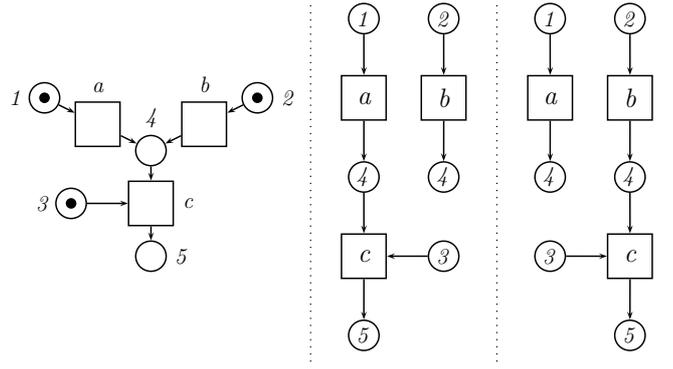

  \begin{center}
    \begin{petrinet}(12,7)
      \Ql(0.5,5):p:1;
      \Q(4.5,5):q:2;
      \ut(1.5,4.5):a::a;
      \ut(3.5,4.5):b::b;
      \qt(2.5,4):r:4;
      \Ql(1,3):cs:3;
      \u(2.5,3):c::c;
      \q(2.5,2):s:5;

      \a p->a; \a q->b; \a a->r; \a b->r;
      \a r->c; \a cs->c; \a c->s;

      \psline[linestyle=dotted](5.5,0)(5.5,7)

      \s(6.5,6.5):p2:1;
      \s(8.0,6.5):q2:2;
      \t(6.5,5):a2:a;
      \t(8.0,5):b2:b;
      \s(6.5,3.5):r2:4;
      \s(8.0,3.5):r2p:4;
      \t(6.5,2):c2:c;
      \s(8.0,2):cs2:3;
      \s(6.5,0.5):s2:5;

      \a p2->a2; \a q2->b2; \a a2->r2; \a b2->r2p; \a r2->c2; \a cs2->c2; \a c2->s2;

      \psline[linestyle=dotted](9.0,0)(9.0,7)

      \s(10.0,6.5):p3:1;
      \s(11.5,6.5):q3:2;
      \t(10.0,5):a3:a;
      \t(11.5,5):b3:b;
      \s(10.0,3.5):r3p:4;
      \s(11.5,3.5):r3:4;
      \t(11.5,2):c3:c;
      \s(10.0,2):cs3:3;
      \s(11.5,0.5):s3:5;

      \a p3->a3; \a q3->b3; \a a3->r3p; \a b3->r3; \a r3->c3; \a cs3->c3; \a c3->s3;

    \end{petrinet}
  \end{center}
  \vspace{-2ex}
  \caption{A net $N$ with its two maximal GR-processes. The process mappings are indicated by labels.}
  \label{fig-unsafe}
\end{figure}

As already described in Section \ref{sec-intro}, a possible strategy
to achieve a more abstract notion of process is to introduce a
suitable equivalence notion, identifying processes which only differ
with respect to the choices of tokens removed from the same place,
thus identifying for example the two processes in \reffig{unsafe}. A
candidate for such an equivalence was proposed in
\cite{best87both}. It is defined by first introducing a simple
transformation on GR-processes; it allows to change causalities in a
process by swapping outgoing arrows between places corresponding to
the same place in the system net. By reflexive and transitive closure
this yields an equivalence notion on finite GR-processes. Slightly
deviating from \cite{best87both}, we define the equivalence for
infinite processes via their finite approximations.

\begin{define}{
  Let $\PP = ((\SS, \TT, \FF, \MM_0), \pi)$ be a process and
  let $p, q \in \SS$ with $(p,q) \notin \FF^+\cup (\FF^{-1})^+$ and
  $\pi(p) = \pi(q)$.
  }
\label{df-swap}
\item[]
  Then $\swap(\PP, p, q)$ is defined as $((\SS, \TT, \FF', \MM_0), \pi)$ with
  \begin{equation*}
    \FF'(x, y) = \begin{cases}
      \FF(q, y) & \text{ iff } x = p,\, y \in \TT\\
      \FF(p, y) & \text{ iff } x = q,\, y \in \TT\\
      \FF(x, y) & \text{ otherwise. }
    \end{cases}
  \end{equation*}
\end{define}

\begin{define}{}\label{df-swapeq}
\item
  Two processes $\PP$ and $\QQ$ of the same net are
  \defitem{one step swapping equivalent} ($\PP \swapeq \QQ$) iff
  $\swap(\PP, p, q)$ is isomorphic to $\QQ$ for some places $p$ and $q$.

\item
We write $\swapeq^*$ for the reflexive and transitive closure of $\swapeq$,
and $\BD{\PP}$ for the $\swapeq^*$-equivalence class of a finite process $\PP$.
The prefix relation $\leq$ between processes is lifted to such
equivalence classes by $\BD{\PP'} \leq \BD{\PP}$ iff\pagebreak[3]
$\PP' \swapeq^* \QQ' \leq \QQ \swapeq^* \PP$
  for some $\QQ',\QQ$.\footnote{It is not hard to verify that if 
  $\PP \swapeq^* \QQ \leq \QQ'$ then $\PP \leq \PP' \swapeq^* \QQ'$ for
  some process $\PP'$. This implies that $\leq$ is a partial order on
  $\swapeq^*$-equivalence classes of finite processes. Alternatively,
  this conclusion will follow from \refthm{orderpreserving}.}

\item
  Two processes $\PP$ and $\QQ$ are \defitem{swapping equivalent}
  ($\PP \swapeq^\infty \QQ$) iff\vspace{-1em}
  \begin{equation*}
    \begin{split}
  &{\downarrow(\{\BD{\PP'}\mid \PP'\leq \PP,~\PP'~\mbox{finite}\})}
  =\\&{\downarrow(\{\BD{\QQ'}\mid \QQ'\leq \QQ,~\QQ'~\mbox{finite}\})}
    \end{split}
  \end{equation*}
  where $\downarrow$ denotes prefix-closure under $\leq$.

\item
$\!$We call a $\swapeq^\infty$\hspace{-2pt}-equivalence class of processes
a \defitem{{\small BD-}process}.
\end{define}

\noindent
Our definition of $\swapeq^\infty$ deviates from the definition of
$\equiv_1^\infty$ from \cite{best87both} to make proofs easier later
on.  We conjecture however that the two notions coincide.

Unfortunately, with respect to the intuition that different ways to
resolve a conflict should give rise to different processes, or the
requirement that a \mbox{P\hspace{-1pt}/T} system should have exactly one maximal process
iff it is conflict-free \cite{goltz86howmany}, the swapping
equivalence relation is too large. Consider for example the net
depicted in \reffig{badswapping}. In the initial situation only two of
the three enabled transitions can fire, which constitutes a conflict.
However, there is only one maximal process up to swapping equivalence.

\begin{figure}[t]
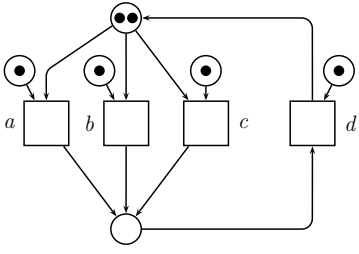

\vspace*{1em}
  \begin{center}
    \begin{petrinet}(12,5)
      \P(3.5,4):pa;
      \P(5,4):pb;
      \P(7,4):pc;
      \P(9.5,4):pd;

      \ul(4,3):a::a;
      \ul(5.5,3):b::b;
      \u(7,3):c::c;
      \u(9,3):d::d;

      \p(5.5,5):p;
      \p(5.5,1):q;

      \av p[210]-(4,4)->[90]a; \a pa->a; \a a->q;
      \a p->b; \a pb->b; \a b->q;
      \a p->c; \a pc->c; \a c->q;
      \av q[0]-(9,1)->[270]d; \a pd->d; \av d[90]-(9,5)->[0]p;

      \pscircle*(5.38,5){0.1}
      \pscircle*(5.62,5){0.1}
    \end{petrinet}
  \end{center}
  \vspace{-5ex}
  \caption{A net with only a single process up to swapping equivalence.$\!$}
  \label{fig-badswapping}
\end{figure}

This example has been known for quite some time
\cite{ochmanski89personal,DMM96,barylska09nonviolence}.
However, we are not aware of a solution, i.e.\ any formalisation of
the concept of a run of a net which correctly represents both
causality and parallelism of the net, and meets the above requirement.
For \defitem{one-safe} nets, i.e.\ nets where places will never carry
more than one token in all reachable markings, as well as for
condition/event systems, (GR-)processes up to isomorphism are
already known to constitute a fully satisfying solution in the above sense. In
this paper we will define a larger subclass of \mbox{P\hspace{-1pt}/T} systems, including
the net of \reffig{unsafe}, on which BD-processes form a satisfying solution.

\section{Conflicts in Place/Transition Systems}\label{sec-conflict}

\noindent
Since we desire an abstract notion of process with the property that
a \mbox{P\hspace{-1pt}/T}-system has exactly one maximal abstract process
iff it is conflict-free, it is essential to have a firm definition of
conflict.
Conflict is a basic notion in the theory of Petri nets, with an
\pagebreak[3]
easy and clear interpretation in one-safe \mbox{P\hspace{-1pt}/T} systems.
Two transitions are in (structural or syntactic) conflict if they
share a common preplace \cite{genrich80dictionary}.
In one-safe nets this coincides with a semantic notion of conflict: if two
transitions share a common preplace and they are both enabled,
only one of them may fire in the next step. Its firing will (at least
temporary) disable the other transition.

In general \mbox{P\hspace{-1pt}/T} systems, the situation concerning
conflicts is more complicated \cite{goltz86howmany}.
First consider the net of \reffig{unsafe}.  In the individual token
interpretation of nets, one could postulate that there is a conflict
between transition $c$ consuming the token produced by $a$ or by
$b$.  Under such an interpretation of conflict, the two maximal
GR-processes match our expectations exactly. Under the collective
token interpretation used in this paper, on the other hand, we
consider this net to be conflict-free, and thus expect only one
maximal process.

Next consider the net in \reffig{badswapping}. In the marking shown,
there are three enabled transitions sharing a preplace. Any pair of
two of them may fire concurrently (even though they share a preplace),
but not all three of them.  This ought to be seen as a conflict. Yet,
if there would be three tokens in the top-most place of that net, the net
would be conflict-free.  This shows that conflict is a more involved
notion here that may no longer be characterised structurally or
syntactically.  In \cite{goltz86howmany}, it was observed that the
traditional definition of conflict covered conflicts between two
transitions only, and the following definition of conflict in general
\mbox{P\hspace{-1pt}/T} systems was proposed.

\begin{define}{
  Let $N \mathbin= (S, T, F, M_0)$ be a net and $M \in
  \bbbn^S\!$.
}\label{df-semanticconflict}
\item
  A finite, non-empty multiset $G \in \bbbn^T$ is in
  \defitem{(semantic) conflict} in $M$ iff
  $
  (\forall t \in G. M\goesto[G \restrictedto \{t\}]) \wedge \neg M\goesto[G]$.

\item
  $N$ is \defitem{(semantic) conflict-free} iff
  no finite, non-empty multiset $G \in \bbbn^T$ is in semantic conflict in any
  $M$ with $M_0 \goesto[] M$.
\end{define}

\begin{trivlist}
\item[\hspace{\labelsep}\bf Remark:]
In a net $(S,T,F,M_0)$ with $S=\{s\}$, $T=\{t,u\}$, $M_0(s)=1$ and
$F(s,t)=F(s,u)=1$, the multiset $\{t,t\}$ is not enabled in $M_0$.
For this reason the multiset $\{t,t,u\}$ does not count as being in
conflict in $M_0$, even though it is not enabled. However, its subset
$\{t,u\}$ is in conflict.
\end{trivlist}

\noindent
We now propose a class of \mbox{P\hspace{-1pt}/T} systems where the
structural definition of conflict matches the semantic definition of
conflict as given above. We require that two transitions sharing a
preplace will never occur both in one step.

\begin{define}{
  Let $N \mathbin= (S, T, F, M_0)$ be a net.
}\label{df-structuralconflict}
\item[]
  $N$ is a \defitem{structural conflict net} iff
  $\forall t, u.
    (M_0 \goesto[]\goesto[\{t, u\}]) \implies
    \precond{t} \cap \precond{u} = \emptyset$.
\end{define}
Note that this excludes self-concurrency from the possible behaviours
in a structural conflict net: as in our setting every transition has
at least one preplace, $t = u$ implies $\precond{t} \cap \precond{u}
\ne \emptyset$.
Also note that in a structural conflict net a non-empty, finite
multiset $G$ is in conflict in a marking $M$ iff $G$ is a set and two distinct
transitions in $G$ are in conflict in $M$.

We will show that the  problem outlined in
Section \ref{sec-semantics}, namely that the transitive closure of the swapping relation equates processes which we would like to distinguish, vanishes for the class of
structural conflict nets. However, the proof of this result is not straightforward. In order to achieve this result, we first introduce the alternative characterisation of BD-processes from \cite{best87both} in terms of an equivalence notion on firing sequences in Section \ref{sec-finiteruns} and then characterise problematic situations and prove an appropriate result in terms of this alternative behaviour description in Section \ref{sec-conflictruns}.

\section{Abstract Runs of Place/Transition Systems}\label{sec-finiteruns}

\noindent
This section is largely based on \cite{best87both}, but with adapted
notation and terminology, and a different treatment of infinite runs.
We recall and reformulate these results in order to use them in the
following two sections.

The behaviour of a net can be described not only by its processes, but
also by its firing sequences. Firing sequences however
impose a total order on transition firings, thereby abstracting from
information on causal dependence, or concurrency, between transition
firings.  To retrieve this information we introduce an
\defitem{adjacency} relation on firing sequences with the intuition
that adjacent firing sequences represent the same run of the net.  We
then define \defitem{FS-runs} in terms of the resulting equivalence
classes of firing sequences.  Adjacency is similar to the idea of
Mazurkiewicz traces \cite{mazurkiewicz95tracetheory}, allowing to
exchange concurrent transitions.  However, it is based on the semantic
notion of concurrency instead of the global syntactic independence
relation in trace theory, similar as in the approach of generalising
trace theory in \cite{HKT95}.

\begin{define}{
  Let $N = (S, T, F, M_0)$ be a net, and $\sigma, \rho \in \FS(N)$.
}\label{df-connectedto}
\item
  $\sigma$ and $\rho$ are \defitem{adjacent}, $\sigma \leftrightarrow \rho$,
  iff $\sigma = \sigma_1 t u \sigma_2$, $\rho = \sigma_1 u t \sigma_2$ and
  $M_0 \goesto[\sigma_1]\goesto[\{t, u\}]$.

\item
  We write $\connectedto$ for the reflexive and transitive closure of $\adjacent$,
  and $[\sigma]$ for the $\connectedto$-equivalence class of a firing sequence $\sigma$.
\end{define}

\noindent
Note that $\connectedto$-related firing sequences contain the same
(finite) multiset of transition occurrences.
When writing $\sigma \connectedto \rho$ we implicitly claim that $\sigma, \rho \in \FS(N)$.
Furthermore $\sigma \connectedto \rho \wedge \sigma\mu \in \FS(N)$
implies $\sigma \mu \connectedto \rho \mu$ for all $\mu \in T^*$.

The following definition introduces the notion of \defitem{partial}
FS-run which is a formalisation of the intuitive concept of a finite,
partial run of a net.

\begin{define}{
  Let $N$ be a net and $\sigma, \rho \in \FS(N)$.
}\label{df-partialrun}
\item
  A \defitem{partial FS-run} of $N$ is an $\connectedto$-equivalence class
  of firing sequences.

\item
  A partial FS-run $[\sigma]$ is a \defitem{prefix} of another partial FS-run
  $[\rho]$, notation \mbox{$[\sigma]\leq[\rho]$}, iff
  $\exists \mu. \sigma \leq \mu \connectedto \rho$.
\end{define}
Note that
$\rho' \connectedto \rho \leq \mu \connectedto \sigma$ implies
$\exists \mu'. \rho' \leq \mu' \connectedto \mu$,
thus the notion of prefix is well-defined, and a partial order.

The following concept of an FS-run extends the notion of a partial
FS-run to possibly infinite runs, in such a way that an FS-run is
\pagebreak[2]
completely determined by its finite approximations.

\begin{define}{
  Let $N$ be a net.
}\label{df-run}
\item[]
  An \defitem{FS-run} $R$ of $N$ is a set of partial FS-runs of $N$ such that
  \begin{itemise}
    \item $[\rho] \leq [\sigma] \in R \implies [\rho] \in R$ ($R$ is prefix-closed), and
    \item $[\sigma], [\rho] \in R \implies \exists [\mu] \in R. [\sigma] \leq [\mu] \wedge [\rho] \leq [\mu]$ ($R$ is directed).
  \end{itemise}
\end{define}
The class of partial FS-runs and the finite elements (in the set
theoretical sense) in the class of FS-runs are in bijective correspondence.
Every finite FS-run $R$ must have a largest element, say
$[\sigma]$, and the set of all prefixes of $[\sigma]$ is $R$.
Conversely, the set of prefixes of a partial FS-run $[\sigma]$  
is a finite FS-run of which the largest element is again $[\sigma]$.

Similar to the construction of FS-runs as sets of equivalence classes of firing
sequences, we define \defitem{BD-runs} as sets of swapping equivalence classes
of finite GR-processes.
There is a close relationship between BD-runs and BD-processes, some
details of which we will give in \refsec{results}.

\begin{define}{
  Let $N$ be a net.
}\label{df-bdrun}
\item
  A \defitem{partial BD-run} of $N$ is a $\swapeq^*$-equivalence class
  of finite processes.

\item
  A \defitem{BD-run} of $N$ is a prefix-closed and directed set of
  partial BD-runs of $N$.
\end{define}
There is a bijective correspondence between partial BD-runs and the
finite elements in the class of BD-runs, just as in the case of FS-runs above.

Much more interesting however is the bijective correspondence between BD-runs
and FS-runs we will now establish. In particular, it allows us to prove
theorems about firing sequences and lift them to processes with relative ease.

\begin{define}{
  Let $N = (S, T, F, M_0)$ be a net, and let $\PP \mathbin= ((\SS, \TT, \FF, \MM_0), \pi)$ be a
  finite process of $N$ and $\sigma \mathbin\in \FS(N)$.
}\label{df-linset}
\item
  $\Lin(\PP) := \{\pi(t_1)\pi(t_2) \ldots \pi(t_n) \mid t_i \in \TT \wedge
  n = |\TT|
  \wedge
  t_i \FF^* t_j \implies i \leq j
  \}$ (the \defitem{linearisations} of $P$).
\item
  $\Pi(\sigma) := \{\PP \mid \sigma \in \Lin(\PP)\}$.
\end{define}
For one-safe nets, $\Pi(\sigma)$ contains exactly one process up to
isomorphism, for any firing sequence $\sigma$ \cite{best87both}.

\begin{theorem}\rm\label{thm-commonimplconnected}
  Let $N \mathbin{=} (S, T, F, M_0)$ be a net, $\sigma, \rho \in \FS(N)$,
  and $\PP, \QQ$ two finite processes of $N$.
  \begin{enumerate}
    \item
      If $\Pi(\sigma) \cap \Pi(\rho) \ne \varnothing$ then
      $\sigma \connectedto \rho$.
    \item
      If $\Lin(\PP) \cap \Lin(\QQ) \ne \varnothing$ then
      $\PP \swapeq^* \QQ$.
  \end{enumerate}
\end{theorem}
\begin{proof}
  See \cite{best87both}.
\qed
\end{proof}

\begin{theorem}\rm\label{thm-adjacentimplcommon}
  Let $N \mathbin{=} (S, T, F, M_0)$ be a net, $\sigma, \rho \in \FS(N)$,
  and $\PP, \QQ$ two finite processes of $N$.
  \begin{enumerate}
    \item
      If $\sigma \adjacent \rho$ then $\Pi(\sigma) \cap \Pi(\rho) \ne \varnothing$.
    \item
      If $\PP \swapeq \QQ$ then $\Lin(\PP) \cap \Lin(\QQ) \ne \varnothing$.
  \end{enumerate}
\end{theorem}
\begin{proof}
  See \cite{best87both}.
\qed
\end{proof}

\begin{theorem}\rm\label{thm-procseqbijection}
  Let $N = (S, T, F, M_0)$ be a net, $\PP,\QQ$ two finite processes of $N$,
  $\sigma\in\Lin(\PP)$, and $\rho\in\Lin(\QQ)$.

  $\sigma \connectedto \rho$ iff $\PP \swapeq^* \QQ$.
\end{theorem}
\begin{proof}
  ``$\Rightarrow$'': We show that
  $\forall n \in \bbbn. (\sigma \adjacent^n \rho \implies \PP \swapeq^* \QQ)$
  by induction on $n$.
  To start, $\sigma \adjacent^0 \rho$ means $\sigma = \rho$, so
  $\sigma \in \Lin(\PP) \cap \Lin(\QQ)$. By \refthm{commonimplconnected} then
  $\PP \swapeq^* \QQ$.
  For the induction step, we need to show that
  $\sigma \adjacent^n \rho \implies \PP \swapeq^* \QQ$.
  There must exists some $\mu$ such that
  $\sigma \adjacent \mu \adjacent^{(n-1)} \rho$.
  By \refthm{adjacentimplcommon} there is some
  $\PP' \in \Pi(\sigma) \cap \Pi(\mu)$.
  So $\sigma \in \Lin(\PP) \cap \Lin(\PP')$
  and per \refthm{commonimplconnected},
  $\PP \swapeq^* \PP'$. That
  $\mu \adjacent^{(n-1)} \rho \implies \PP' \swapeq^* \QQ$
  follows from the induction assumption.

  ``$\Leftarrow$'': Goes likewise but with the r\^oles of
  $\adjacent$ and $\swapeq$ and those of $\Pi$ and $\Lin$ exchanged.
\qed
\end{proof}
The functions $\Lin$ and $\Pi$ can be lifted to equivalence classes of
finite processes and firing sequences, respectively, by

$\Lin(\BD{\PP}) := [\sigma]$ for $\sigma$ an arbitrary element of $\Lin(\PP)$, and

$\Pi([\sigma]) := \BD{\PP}$ for $\PP$ an arbitrary element of $\Pi(\sigma)$.

\noindent
\refthm{procseqbijection} ensures that these liftings are
well-defined, and that they are inverses of each other, thereby
obtaining a bijective correspondence between partial BD-runs and
partial FS-runs. The following theorem tells that this bijection
respects the prefix ordering between runs.

\begin{theorem}\rm\label{thm-orderpreserving}
  Let $N = (S, T, F, M_0)$ be a net, $\PP,\QQ$ two finite processes of $N$,
  $\sigma\in\Lin(\PP)$, and $\rho\in\Lin(\QQ)$.

  $[\sigma] \leq [\rho]$ iff $\BD{\PP} \leq \BD{\QQ}$.
\end{theorem}

\begin{proof}
``$\Leftarrow$'': Take $\PP'\in\BD{\PP}$ and $\QQ'\in\BD{\QQ}$ such that $\PP' \leq \QQ'$.
If follows immediately from Definitions~\ref{df-extension} and~\ref{df-linset}
that any $\sigma'\in\Lin(\PP')$ can be extended to some
$\rho'\in\Lin(\QQ')$, so that $[\sigma]=[\sigma']\leq[\rho']=[\rho]$.

``$\Rightarrow$'': Take $\rho'\inp[\rho]$ such that $\sigma \leq \rho'$ and take
$\QQ'=((\SS',\TT',\FF',\MM_0'),\pi')\inp\Pi(\rho')$. By \refdf{linset},
$\TT'$ can be 
enumerated as $t_1 t_1 \ldots t_n$ such that $  t_i \FF'^* t_j \implies
i \leq j$ and $\rho'=\pi'(t_1)\pi'(t_2) \ldots \pi'(t_n)$.
So $\sigma=\pi'(t_1)\pi'(t_2) \ldots \pi'(t_m)$ with $m \leq n$.
It follows from the remark below \refdf{extension} that
$\QQ'$ has a prefix $\PP'$ with transitions $\{t_1,\ldots,t_m\}$ such that
$\PP'\in\Pi(\sigma)$. Hence $\BD{\PP}=\BD{\PP'}\leq\BD{\QQ'}=\BD{\QQ}$.
\qed
\end{proof}
Since BD-runs are created out of the ordered space of partial BD-runs
of a net in the same way as FS-runs are created out of partial
FS-runs, this immediately yields a bijective correspondence also between
infinite (in the set-theoretical sense) BD-runs and infinite FS-runs.
This bijection respects the subset relation $\subseteq$ between runs,
which is the counterpart of the prefix relation $\leq$ between partial
runs, and hence also the concept of a maximal run.

\section{Abstract Runs of Structural Conflict Nets}\label{sec-conflictruns}

\noindent
This section formally uses FS-runs; however the results carry over
to BD-runs easily, via the bijection established in \refsec{finiteruns}.

Returning to the example of \reffig{badswapping},
we find that the depicted net has only one maximal FS-run:
$[abdc] \mathbin= [adbc] \mathbin= [adcb] = [acdb] \mathbin=
[cadb] \mathbin= [cdab] \mathbin= [cdba] \mathbin= [cbda] \mathbin=
[bcda] \mathbin= [bdca] \mathbin= [bdac] \mathbin= [badc]$.
The conflict between the initially enabled sets of transitions $\{a, b\}$, $\{b, c\}$,
and $\{a, c\}$ has not been resolved; rather all
possibilities have been included in the same run.
The following definition describes runs for which this is
not the case.

\begin{define}{
  Let $N = (S, T, F, M_0)$ be a net.
}\label{df-conflictfreerun}
\item[]
  An FS-run $R$ is \defitem{conflict-free} iff
  for all finite, non-empty multisets $G \inp \bbbn^T$ and all $\sigma \inp T^*$
  $$
  (\forall t \in G. [\sigma t^{G(t)}] \in R \wedge
  M_0 \goesto[\sigma] \goesto[G \restrictedto \{t\}]) \implies
  M_0 \goesto[\sigma] \goesto[G]\trail{.}
  $$
\end{define}

\noindent
We will now show that in structural conflict nets every run is conflict free.
For structural conflict nets thus holds what one would intuitively expect:
every conflict in the net gives rise to distinct runs, each one representing a
particular way to resolve the conflict.

\begin{theorem}\rm\label{thm-conflictfreerun}
  Let $N$ be a structural conflict net.

  Every FS-run $R$ of $N$ is conflict-free.
\end{theorem}
\begin{proof}
  Let $R$ be an FS-run of $N = (S, T, F, M_0)$, $\sigma \in T^*$,
  and $G \in \bbbn^T$ a finite, non-empty multiset such that
  $\forall t \inp G. [\sigma t^{G(t)}] \inp R \wedge
  M_0 \goesto[\sigma]\goesto[G \restrictedto \{t\}]$.
  Let $M$ be the unique marking of $N$ with $M_0 \goesto[\sigma] M$.
  We have to show that $M \goesto[G]$.

  No transition $t$ can occur more than once in $G$ as self-concurrency cannot
  occur in structural conflict nets and
  $M_0 \goesto[\sigma]\goesto[G \restrictedto \{t\}]$.

  Let $t, u \in G$, $t \ne u$. Then $M \goesto[t]{} \wedge M \goesto[u]$.
  Since $R$ is directed, there exist $\rho, \mu \in T^*$ with
  $\sigma t \rho \connectedto \sigma u \mu$.
  By \refdf{connectedto}, $\sigma t \rho$ and $\sigma u \mu$ must
  contain the same multiset of transitions.
  Hence somewhere in the sequence
  $\sigma t \rho = \nu_1 \adjacent \nu_2 \adjacent \cdots \adjacent \nu_n
  = \sigma u \mu$ the transitions $t$ and $u$ must be exchanged, i.e.\
  $\nu_i = \nu' t u \nu'' \adjacent \nu' u t \nu'' = \nu_{i+1}$.
  Thus there is a marking $M'$ with $M_0 \goesto[\nu'] M' \goesto[\{t, u\}]$.
  Since $N$ is a structural conflict net,
  $\precond{t} \cap \precond{u} = \varnothing$.
  As this holds for all $t, u \in G$, it follows that $M \goesto[G]$.
\qed
\end{proof}

\begin{theorem}\rm\label{thm-conflictfree}
  Let $N$ be a structural conflict net.

  If $N$ has exactly one maximal FS-run then $N$ is conflict-free.
\end{theorem}
\begin{proof}
  Let $N = (S, T, F, M_0)$.
  Assume $N$ has a conflict, i.e.\ there exists $\sigma \in T^*$, $M \in \bbbn^S$,
  $G \in \bbbn^T$, $G$ finite, with $M_0 \goesto[\sigma] M$, $\neg M
  \goesto[G]$ and $\forall t \in G. M \goesto[G \restrictedto \{t\}]$.
  We show that $N$ has no unique maximal FS-run.

  For every $t \in G$, the set $\{[\rho] \mid [\rho] \leq [\sigma t^{G(t)}]\}$
  constitutes an FS-run of $N$.
  Hence, a unique maximal FS-run of $N$ would be a superset of
  $\{[\sigma t^{G(t)}] \mid t \in G\}$, 
  and thus not conflict-free. However, every FS-run of $N$ must be
  conflict-free according to \refthm{conflictfreerun}.
\qed
\end{proof}

\section{BD-Processes of Structural Conflict Nets}\label{sec-results}

\noindent
In this section we adapt \refthm{conflictfree} from runs
to BD-processes, i.e.\ GR-processes up to $\swapeq^\infty$.
To this end, we give a mapping from GR-processes to BD-runs.

\begin{define}{
  Let $N$ be a net and $\PP$ a process thereof.
}\label{def-bdify}
\item[]
  Then $\BDify{\PP} := \downarrow \{\BD{\PP'} \mid
    \PP' \leq \PP,~ \PP' \text{ finite}\}$.
\end{define}
Note that, by \refdf{swapeq}, $\PP \swapeq^\infty \QQ$ iff
$\BDify{\PP} = \BDify{\QQ}$.

\begin{lemma}\rm\label{lem-bdprocessisrun}
  Let $N$ be a net and $\PP$ a process thereof.

  $\BDify{\PP}$ is a BD-run.
\end{lemma}
\begin{proof}
  Prefix-closure holds by definition of $\downarrow$, using the
  transitivity of $\leq$.

  For any ordered set $X$, if $X$ is directed, so is \mbox{$\downarrow\! X$}.
  Hence it suffices to show that $ \{\BD{\PP'} \mid \PP' \mathbin\leq \PP,\, \PP' \text{
  finite}\}$ is directed.
  Take finite $\PP_i := ((\SS_i, \TT_i, \FF_i, {\MM_0}_i), \pi_i) \leq \PP$
  for $i \mathbin= 1,2$.
  Then $\PP' := ((\SS_1 \cup \SS_2, \TT_1 \cup \TT_2, \FF_1 \cup \FF_2,
    {\MM_0}_1 \cup {\MM_0}_2),\linebreak[3] \pi_1 \cup \pi_2)\\ \leq \PP$
  and $\PP'$ is finite.
  Moreover, $\PP_i \leq \PP'$ for $i\mathbin=1,2$ and thus $\BD{\PP_i} \leq \BD{\PP'}$.
  Hence $\BDify{\PP}$ is directed.
\qed
\end{proof}

\begin{IPL}
\noindent
We now show that the existence of a unique maximal GR-process implies the
existence of a unique maximal BD-run.

\begin{lemma}\rm\label{lem-towardsinfinity}
  Let $N$ be a net.

  Every process $\PP$ of $N$ is a prefix of a maximal process of $N$.
\end{lemma}
\begin{proof}
  The set of all processes of $N$ of which $\PP$ is a prefix
  is partially ordered by $\leq$.
  Every chain in this set has an upper bound, obtained
  by componentwise union.
  Via Zorn's Lemma this set contains at least one maximal process.
  \qed
\end{proof}

\begin{lemma}\rm\label{lem-onemaximalproc}
  Let $N$ be a net.

  If $N$ has exactly one maximal GR-process up to $\swapeq^\infty$
  then $N$ has exactly one maximal BD-run.
\end{lemma}
\begin{proof}
  Take any finite processes $\PP, \PP'$ of $N$.
  According to \reflem{towardsinfinity} there are maximal
  processes $\QQ, \QQ'$ of $N$ with \mbox{$\PP \leq \QQ$}, $\PP' \leq \QQ'$.
  As $N$ has exactly one maximal process up to $\swapeq^\infty$,
  $\QQ \swapeq^\infty \QQ'$ and
  $\BDify{\QQ} = \BDify{\QQ'}$. Hence as $\BD{\PP'} \in \BDify{\QQ'}$ also
  $\BD{\PP'} \in \BDify{\QQ}$. Since $\BDify{\QQ}$ is directed,
  there exists a $\QQ''$ with
  $\BD{\PP} \leq \BD{\QQ''} \wedge \BD{\PP'} \leq \BD{\QQ''}$.
  As this holds for any finite processes $\PP, \PP'$ the set of all
  equivalence classes of finite processes of $N$ is directed and hence a
  BD-run.
  Naturally this is the largest BD-run.
\qed
\end{proof}

\noindent
We can now conclude our main result: A semantic conflict in structural
conflict nets generates multiple maximal GR-processes even up to
swapping equivalence.

\begin{corollary}\rm\label{cor-structuralconflictbdprocs}
  Let $N$ be a structural conflict net.

  If $N$ has only one maximal GR-process up to $\swapeq^\infty$ then $N$ is conflict-free.
\end{corollary}
\begin{proof}
  This follows directly from \reflem{onemaximalproc} and
  \refthm{conflictfree}, using the bijection between FS-runs and
  BD-runs of \refsec{finiteruns}.
\qed
\end{proof}

\noindent
It would be interesting to show the reverse direction of
\refcor{structuralconflictbdprocs}, i.e.\ to prove that a structural
conflict net has exactly one maximal GR-process up to $\swapeq^\infty$ {\it iff} it is conflict-free.
We do conjecture that this holds for countable nets.
Even for processes generated by finite nets though,
we find it difficult to apply a similar proof technique by establishing the
necessary bijective correspondence between infinite BD-processes and infinite BD-runs.
\end{IPL}

\bibliographystyle{alpha}

\end{document}